\documentclass[12pt, DIV=9]{article}
\usepackage[utf8]{inputenc}
\usepackage[english]{babel}
\usepackage{amsmath}
\usepackage{amsthm}

\usepackage{soul} 
\usepackage{cancel}

\newtheorem{theorem}{Theorem}

\numberwithin{theorem}{section}
\newtheorem{hypo}[theorem]{Hypothesis}
\newtheorem{lemma}[theorem]{Lemma}

\newtheorem{prop}[theorem]{Proposition}
\newtheorem{cor}[theorem]{Corollary}
\newtheorem{announcement}[theorem]{Announcement}
\theoremstyle{definition}
\newtheorem{definition}[theorem]{Definition}
\newtheorem{bem}[theorem]{Remark}

\numberwithin{equation}{section}

\usepackage{tikz}
\usetikzlibrary{arrows}
\usetikzlibrary{positioning}
\usepackage{color} 
\usepackage{amssymb}
\usepackage{graphicx}
\usepackage{pdflscape}
\usepackage{hyperref}
\usepackage{wrapfig}
\usepackage{textcomp}
\usepackage{mathtools}
\usepackage{braket}

\usepackage{dsfont}

\newcommand{\oline}{\overline}   
\newcommand{\uline}{\underline}   
\newcommand{\Om}{\Omega}                

\newcommand{\la}{\langle}
\newcommand{\ra}{\rangle}

\newcommand{\boE}{\mathbf{E}}         
\newcommand{\boI}{\mathbf{I}}         

\newcommand{\bfone}{\mathbf{1}}
\newcommand{\cA}{\mathcal{A}}
\newcommand{\cB}{\mathcal{B}}

\newcommand{\cF}{\mathcal{F}}

\newcommand{\cI}{\mathcal{I}}

\newcommand{\cR}{\mathcal{R}}
\newcommand{\cS}{\mathcal{S}}
\newcommand{\cT}{\mathcal{T}}
\newcommand{\cU}{\mathcal{U}}

\newcommand{\cW}{\mathcal{W}}

\newcommand{\field}[1]{\mathds{#1}}
 
\newcommand{\bbone}{\field{1}} 
 
\newcommand{\HH}{\field{H}}

\newcommand{\RR}{\field{R}}     
\newcommand{\NN}{\field{N}}     
\newcommand{\CC}{\field{C}}     
\newcommand{\bbW}{\field{W}}     
\newcommand{\fD}{\mathfrak{D}}
\newcommand{\fE}{\mathfrak{E}}
\newcommand{\fF}{\mathfrak{F}}
\newcommand{\fH}{\mathfrak{H}}  
\newcommand{\fd}{\mathfrak{d}}  

\newcommand{\fh}{\mathfrak{h}}

\newcommand{\bD}{{\overline D}}  
\newcommand{\bP}{{\overline P}}  
\newcommand{\bX}{{\overline X}}

\newcommand{\bchi}{{\overline{\chi}}}
\newcommand{\btheta}{\bar{\theta}}
\newcommand{\ur}{{\underline r}}
\newcommand{\uv}{{\underline v}}
\newcommand{\uw}{{\underline w}}

\newcommand{\hfD}{\widehat{\fD}}

\newcommand{\hH}{\widehat{H}}   
\newcommand{\hfH}{\widehat{\fH}} 
\newcommand{\hP}{\widehat{P}}    
\newcommand{\hcR}{\widehat{\mathcal{R}}} 
\newcommand{\hT}{\widehat{T}}    
\newcommand{\hW}{\widehat{W}}

\newcommand{\hpsi}{\hat{\psi}}

\newcommand{\tk}{{\widetilde{k}}} 

\newcommand{\tuw}{\widetilde{\uw}}

\newcommand{\va}{{\vec{a}}}
\newcommand{\vb}{{\vec{b}}}

\newcommand{\fin}{\mathrm{fin}}

\newcommand{\red}{\mathrm{red}}

\newcommand{\rRe}{\mathrm{Re}}               
\newcommand{\Ran}{\mathrm{Ran}}              

\newcommand{\cirS}{\mathop{\bigcirc\kern -.73em {\scriptstyle{\rm S}}}}

\newcommand{\dom}{\mathrm{dom}}

\newcommand{\cern}{\mathrm{Ker}} 
\newcommand{\op}{\mathrm{op}}

\newcommand{\ph}{{\mathrm{ph}}}

\newcommand{\rmd}{\mathrm{d}}
\newcommand{\free}{{0,0}}
\newcommand{\gI}{{(I)}}
\newcommand{\qHH}{{\HH}}
\newcommand{\qWW}{{\bbW}}

\begin{document}

\title{The Spectral Renormalization Flow Based on the Smooth
  Feshbach--Schur Map: The Introduction of the Semi-Group Property.}

\author{
Volker Bach \\
\small{IAA, TU Braunschweig, Germany (v.bach@tu-bs.de)}
\and
Miguel Ballesteros \\
\small{IIMAS, UNAM, Mexico (miguel.ballesteros@iimas.unam.mx)} 
\and 
Jakob Geisler \\
\small{IAA, TU Braunschweig, Germany (jakob.geisler@tu-bs.de)}
}
\date{20-Jul-2025}

\maketitle

\textbf{Abstract:} The spectral renormalization method is a powerful
mathematical tool that is prominently used in spectral theory in the
context of low-energy quantum field theory and its original
introduction in \cite{BachFroehlichSigal1998a,
  BachFroehlichSigal1998b} constituted a milestone in the
field. Inspired by physics, this method is usually called
renormalization group, even though it is not a group nor a
semigroup (or, more properly, a flow). It was only in 2015
\cite{BachBallesterosFroehlich2015} when a flow (or semigroup)
structure was first introduced using an innovative definition of the
renormalization of spectral parameters. The spectral renormalization
flow in \cite{BachBallesterosFroehlich2015}, however, is not
compatible with the smooth Feshbach--Schur map (this is stated
as an open problem in \cite{BachBallesterosFroehlich2015}), which is a
lamentable weakness because its smoothness is a key feature that
significantly simplifies the proofs and makes it the preferred tool in
most of the literature. In this paper we solve this open problem
introducing a spectral renormalization flow based on the smooth
Feshbach--Schur map.

\section{Introduction} \label{intro}
%
In this paper, we exhibit new aspects of the smooth Feshbach--Schur
map and draw some conclusions to the operator-theoretic spectral
renormalization that has been constructed to analyze spectra of
Hamiltonians on Fock spaces.

The \textit{smooth Feshbach--Schur map} originates from the
Feshbach projection method, which was introduced in
\cite{Feshbach1958} as a tool in analytic pertubation theory. It is
closely related to the Schur complement \cite{Schur1917} and from a
more general point of view to the Grushin problem
\cite{Grushin1971}. It has been further developed to the Feshbach map
in \cite{BachFroehlichSigal1998b} in order to handle problems in the
spectral analysis of non-isolated eigenvalues and
resonances. Honouring Schur's contribution, we call it the
\textit{Feshbach--Schur map} in this paper. For its definition one
chooses a projection $P$ on some Hilbert space $\fH$. Then an operator
$H$ is mapped to a new operator $\cF_P(H)$\, and the map $H \mapsto
\cF_P(H)$ is called Feshbach--Schur map. It possesses an isospectral
property which, however, in this context does not mean that $\cF_P(H)$
and $H$ have the same spectrum, but rather that $0$ is a spectral
point of $H$ if and only if $0$ is a spectral point of $\cF_P(H)$, as
an operator on $\Ran(P)$, and that the spectral type at $0$ is the
same for both $H$ and $\cF_P(H)$. In particular, $0$ is an eigenvalue
of $H$ if and only if $0$ is an eigenvalue of $\cF_P(H)$, and in this
case the multiplicities agree.

The Feshbach--Schur map has been generalized in
\cite{BachChenFroehlichSigal2003} (see also
\cite{GriesemerHasler2008}) replacing the projections $P$ and $\bP =
\mathbf{1} - P$ by positive operators $0 \leq \chi \leq \mathbf{1}$
and $\bchi = \sqrt{\mathbf{1}-\chi^2}$, respectively. These operators
are usually defined in terms of smooth functions by functional
calculus, with $\chi$ and $\bchi$ being smoothed versions of
characteristic functions. In contrast, (sharp) characteristic
functions give rise to projections and correspond to the original
Feshbach--Schur map. To stress the difference between the two, we call
the latter \textit{sharp} or \textit{projection-based} Feshbach--Schur
map.

The definition of the smooth Feshbach--Schur map additionally requires
the choice of an operator $T$ that commutes with both $\chi$ and
$\bchi$. In the present context, the operator $H$ is the sum of a
\textit{free part, whose spectral properties are assumed to be
  perfectly known and explicitly available,} plus an interaction. It
seems natural to choose the operator $T$ to be this free part, which
is then physically interpreted as the unperturbed energy of the system
under consideration. This common assumption is, however, misleading
because the actual role of $T$ is to solve the mismatch problem
arising from the fact that $\chi$ is not a characteristic function of
the free energy operator, which is the case for the sharp
Feshbach--Schur map. Here we adopt a notation which is consistent with
the fact that operator $T$ is not a fundamental part of the model, but
only a convenient parameter; as opposed to the previous literature, in
this paper we do not regard the operators $H$ and $T$ as a Feshbach
pair but include $T$ as part to the defining parameters of the smooth
Feshbach--Schur map. We stress this new point of view by writing
$F_{\chi,T}(H)$ for the smooth Feshbach--Schur map applied to the
operator $H$. Our formalism introduces innovative features to
construct a new approach to the smooth Feshbach--Schur map which
satisfies a flow (semigroup) property for the first time. This has
been an open problem since 2015, as in
\cite{BachBallesterosFroehlich2015} the first introduction of a
Feshbach--Schur map with a semigroup property was presented, but the
approach crucially depended on using the sharp Feshbach--Schur map. In
the present paper, we characterize the freedom for the choice of $T$
and compare the sharp Feshbach--Schur map to its smooth
counterpart. Concretely, $F_{\chi,T}(H)$ only depends on $T$ on the
range of $\chi \bchi$ (note that $\chi$ and $\bchi$ commute). Hence,
if $\chi$ and $\bchi$ are thought of as smoothed-out versions of
orthogonal spectral projections $P$ and $P^\perp$, this range might
become arbitrarily small. This weak dependence on $T$ allows a far
broader application of the smooth Feshbach--Schur map. We exhibit the
additional freedom in the choice of $T$ by introducing an operator $S$
which is not identical to $T$, but fulfills the same hypothesis and
agrees with $T$ in the overlap region of $\chi$ and $\bchi$. We show
that the smooth Feshbach--Schur maps $F_{\chi,T}(H)$ and
$F_{\chi,S}(H)$ are the same and how $F_{\chi,T}(H)$ can be expressed
in terms of $W_S = H-S$.

As mentioned above, the smooth Feshbach--Schur map is a
generalization of the sharp Feshbach-Schur map in the sense
that, if $\chi$ and $\bchi = \chi^\perp$ are orthogonal projections,
then the smooth Feshbach--Schur map reduces to the
Feshbach--Schur map. A new result of the present paper is that,
conversely, the isospectrality of the smooth Feshbach--Schur map
can be derived from the isospectrality of the sharp
Feshbach--Schur map, provided the latter is set up on a slightly
bigger Hilbert space than the original one.

In \cite{BachChenFroehlichSigal2003}, the smooth Feshbach--Schur is
used to introduce a simpler version of the BFS spectral
renormalization method, and it is further developed in
\cite{BachBallesterosFroehlich2015}. Although the sharp
Feshbach--Schur map and the smooth Feshbach--Schur map might be
equally powerful mathematical tools, the differentiability properties
of the smooth Feshbach--Schur map considerably simplify the
analysis. For this reason, the latter has been the preferred method in
most of the literature. It has been utilized in a variety of problems:
in \cite{BachChenFroehlichSigal2007} and \cite{Chen2008}, it was used
for the renormalization of the electron mass whereas in
\cite{Faupin2008} and \cite{BallesterosFaupinFroehlichSchubnel2015}
the existence of atomic resonances, including the Lamb shift, was
addressed. While life-times of resonances were estimated in
\cite{HaslerHerbstHuber2008} still using the sharp Feshbach--Schur
map, in \cite{GriesemerHasler2008}, the algebraic and analytic
properties of the smooth Feshbach--Schur map were clarified, and the
method was generalized to non-selfadjoint choices for $\chi$ and
$\bchi$. In \cite{GriesemerHasler2009}, analyticity properties of the
ground state in the Standard Model of Non-Relativistic QED were
established using the smooth Feshbach--Schur map. In \cite{Sigal2009},
the existence of resonances in the Standard Model of Non-Relativistic
QED was shown and this for the first time without requiring an
infrared regularization. In \cite{FroehlichGriesemerSigal2011}, the
limiting absorption principle was studied, and in
\cite{HaslerHerbst2011} the existence of ground states in the
Spin-Boson model without an infrared regularization was proved.

Both the smooth and sharp versions of the Feshbach--Schur map are
spectral localization methods satisfying the isospectral property
described above. For the smooth Feshbach--Schur map this
specifically implies that $H-z$ is invertible if and only if
$\cF_{\chi,T-z}(H-z)$, restricted to the range $\Ran(\chi)$ of $\chi$,
is invertible. Hence the spectrum of $H$ corresponds to the points
where (the restriction to $\mathrm{Ran}(\chi)$ of)
$\cF_{\chi,T-z}(H-z)\big|_{\Ran(\chi)}$ is not invertible, whenever
$\cF_{\chi,T-z}(H-z)$ exists. The advantage of the (smooth or sharp)
Feshbach--Schur map is that the domain on which $\cF_{\chi,T-z}(H-z)$
acts excludes the orthogonal complement of $\Ran(\chi)$. This is
frequently rephrased by saying that the \textit{degrees of freedom
  corresponding to the complement of the range of $\chi$ are
  eliminated}. Thus the smaller the range of $\chi$ the better. The
price to pay for the elimination of degrees of freedom is that the set
of spectral points $z$ where the Feschbach-Schur map is defined is the
more reduced, the smaller the range of $\chi$ is, and in general it is
very difficult, if not impossible, to localize these spectral points
if the range of $\chi$ is ``too small''. One of the main applications
of the method is the construction of eigenvalues and resonances. The
first step is to choose an appropriate operator $\chi$ and a first
region $\cU_0 \subset \CC$ where the eigenvalue sought for is
localized and $z \in \cU_0$ is chosen from. A first application of the
smooth Feshbach--Schur map gives more accurate information on
the location of the eigenvalue and allows to choose a new, much
smaller region $\cU_1 \subset \cU_0$ where, for $z \in \cU_1$, a
second application of the smooth Feshbach--Schur map can be
performed. This second application provides a further elimination of
degrees of freedom and a, yet, smaller region $\cU_2 \subset \cU_1$
where the eigenvalue is localized. Proceeding in this vein, more and
more applications of the smooth Feshbach--Schur map generate
ever smaller regions where the eigenvalue can be found. As the number
of applications of the smooth Feshbach--Schur map tends to
infinity, the eigenvalue is reconstructed.

In the case of quantum field theory that we present in this paper, the
operator $\chi$ is a function of the free boson energy operator
$H_\ph$ [see \eqref{free-1}-\eqref{free-2}] and another parameter
$\alpha > 0$. Then $\chi$ is substituted by $\chi_\alpha(H_\ph) \equiv
\chi_\alpha$ [see Definitions~\ref{dfn-smoothfamily} and
  \ref{def:rescaledfsmap}]. In our setting, the range of $\chi_\alpha$
decreases, as $\alpha$ increases, and the set of points $z$ for which
the smooth Feshbach--Schur map is defined gets ever smaller. The
shrinking properties of the ranges of $\chi_\alpha$ and of the set of
spectral parameters is problematic. For this reason, we compose the
smooth Feshbach--Schur map with suitable unitary scaling operators to
compensate for the shrinking. The resulting map is called the
renormalization map and denoted by $\cR_\alpha(H)(z)$. The above
description is valid for every spectral renormalization scheme. In
spite of the fact, however, that these schemes are frequently called
\textit{renormalization group} in the literature, they are not a group
nor even a semigroup or a flow. Only in
\cite{BachBallesterosFroehlich2015} the flow property was established
- but for the sharp Feshbach--Schur map and not for the smooth
Feshbach--Schur map. For this reason, we prefer to call this method
\textit{spectral renormalization}. In the present paper we introduce
the first spectral renormalization scheme which is a flow (or a
semigroup) for the smooth Feshbach--Schur map: In
Theorem~\ref{thm:fullsemigroup}, we prove the flow property for the
spectral renormalization scheme introduced here and show that the
renormalization map obeys
\begin{align}
\cR_{\alpha + \beta}(H) = \cR_\alpha\big( \cR_{\beta}(H) \big),
\end{align}
where we omit the spectral parameter $z$.    

One of the conceptual advantages of the flow property is that in all
previous works it was necessary to construct a sequence of operators
obtained from iterated applications of the renormalization map. The
flow property allows to consider only one application of the
renormalization map for $\alpha$ and take $\alpha $ to infinity. The
eigenvalue whose construction is sketched in the previous paragraphs
belongs to the range of a complex-valued function $E_\alpha$ (the
renormalization of the spectral parameter) that shrinks exponentially,
as $\alpha$ tends to infinity. This provides an approximation of the
eigenvalue with an error decreasing exponentially to zero, as $\alpha$
tends to infinity. Moreover, the renormalization flow that we define
can be constructed for every positive $\alpha$. This allows to take
the derivative with respect to $\alpha$ and obtain a differential
equation that simplifies the formulae. Although a flow was already
derived in \cite{BachBallesterosFroehlich2015}, it was only obtained
for the sharp Feshbach--Schur map. The corresponding result for the
smooth counterpart was regarded an open problem in
\cite{BachBallesterosFroehlich2015}. Here, we solve this problem. The
solution is important because the smooth version is significantly
simpler than the sharp version on a technical level. Moreover, the
smooth version gives rise to a simpler differential equation because
its smoothness allows to take derivatives which are not distributions
(in contrast to the sharp case).

\subsection{Organization of the Paper}

The paper is organized as follows: Section~\ref{short} is a short
version of the paper, in which we describe the mathematical framework
and the main results in a self-contained fashion, omitting proofs and
technicalities. It is itself divided in two parts, namely,
Section~\ref{SFSM} and Section~\ref{RF}. In the former we present our
main results with regard to the theoretical (abstract) study of the
smooth Feshbach--Schur map, the main theorems of this section
being Theorem~\ref{IND} and Theorem~\ref{Tmaintsyeah}. In
Section~\ref{RF}, we present the main results of the paper, i.e., the
flow property of the renormalization map that we define. They are
presented in Theorems~\ref{renspect} and \ref{thm:fullsemigroup}. In
Section~\ref{sec:SFM} we provide the proofs of Section~\ref{SFSM} and
in Subsection~\ref{subsec:sharpsmooth} we derive the isospectrality of
the smooth Feshbach--Schur map from the isospectrality of the
Feshbach--Schur map. In Section~\ref{sec:RG} we derive the proofs of
Section~\ref{RF}. In Section~\ref{Itera}, we announce our forthcoming
results with regard to iterative applications of the renormalization
map.

\section{Mathematical Framework and Main Results} 
\label{short}
%
\subsection{A New Approach to the Smooth Feshbach--Schur Map} 
\label{SFSM}
%
As described above, the smooth Feshbach--Schur map is a powerful
mathematical tool for spectral analysis. We denote by $\fH$ the
underlying Hilbert space. The smooth Feshbach--Schur map
requires three auxiliary operators $\chi$, $\bchi$, and $T$ on $\fH$
for its definition.
%
\begin{hypo} \label{hypo-chi} 
The operators $\chi$ and $\bchi$ are self-adjoint, positive,
bounded, mutually commuting, and additionally obey 
$\chi^2 + \bchi^2 = 1$. 
\end{hypo} 
%
We define
\begin{align} \label{hilbert-1}
\fH_\chi \ := \ \oline{\Ran(\chi)} 
\, , \quad
\fH_\bchi \ := \ \oline{\Ran(\bchi)} 
\, , \quad
\fH_{\chi\bchi} \ := \ \fH_\chi \cap \fH_\bchi \, ,
\end{align}
where $\Ran(A)$ denotes the range of an operator 
$A$ (and, moreover, $\cern(A)$ its kernel). Let $P_\chi$, $P_\bchi$, and
$P_{\chi\bchi}$, respectively, be the orthogonal projections onto these 
closed subspaces, so that
\begin{align} \label{hilbert-2}
\fH_\chi \ := \ P_\chi \fH \, , \quad 
\fH_\bchi \ := \ P_\bchi \fH \, , \quad 
\fH_{\chi\bchi} \ := \ P_{\chi\bchi} \fH \, .
\end{align}   
We note that 
\begin{align} \label{hilbert-2,1}
\chi \ = \ \chi \, P_\chi \ = P_\chi \, \chi
\quad \text{and} \quad 
\bchi \ = \ \bchi \, P_\bchi \ = P_\bchi \, \bchi \, ,
\end{align}   
and we henceforth do not
distinguish $\chi: \fH \to \fH$ and $\bchi: \fH \to \fH$ from 
their restriction to $\fH_\chi$ and to  $\fH_\bchi$, respectively.

\begin{hypo} \label{hypo-T} 
Assume Hypothesis~\ref{hypo-chi}. The operator $T$ is densely defined 
on a subspace $\fD \equiv \fD(T) \subseteq \fH$ and closed. 
It commutes with $\chi$ and $\bchi$ in the sense that 
$P_\chi, P_\bchi, \chi, \bchi: \fD \to \fD$ and
\begin{align} \label{lele-a}
\chi \, T \ \subset \ T \, \chi \, , \quad 
\bchi \, T \ \subset \ T \, \bchi \, , 
\end{align}
hold true. The restricted operator
\begin{align} \label{lele-b}
\text{$T_\bchi$ is bounded invertible on $\fH_\bchi$,}
\end{align}
where $T_\bchi := T\big|_{\fH_\bchi}: \fD \cap \fH_\bchi \to \fH_\bchi$ 
denotes the restriction of $T$ to $\fH_{\bchi}$.
\end{hypo} 
%
Recall that a closed linear operator $A: \fd \to \fh$ 
defined on a dense subspace $\fd \subseteq \fh$ of a Hilbert space $\fh$
is called \textit{bounded invertible}, 
if $A: \fd \to \Ran(A)$ is injective, $\Ran(A) \subseteq \fh$ is dense,
and $\inf_{\psi \in \fd, \|\psi\|=1} \|A \psi \| > 0$. In this case
$\Ran(A) = \fh$,  $A: \fd \to \Ran(A)$ is a bijection, 
and the inverse map $A^{-1}: \fh \to \fd$ defines a bounded 
linear operator on $\fh$. 

Similarly to \eqref{hilbert-2} we define the subspaces
\begin{align} \label{hilbert-3}
\fD_\chi \ := \ P_\chi \fD \ = \ \fH_\chi \cap \fD   
\quad \text{and} \quad
\fD_\bchi \ := \ P_\bchi \fD \ = \ \fH_\bchi \cap \fD \, . 
\end{align}   
Then the restriction of $T$ to $\fH_\bchi$ reads
$T_\bchi: \fD_\bchi \to \fH_\bchi$, and we similarly 
denote by $T_\chi := T\big|_{\fH_\chi}: \fD_\chi \to \fH_\chi$ 
the restriction of $T$ to $\fH_\bchi$. We observe that, since $T$ is
closed, so are $T_\chi$ and $T_\bchi$.

Assuming Hypotheses~\ref{hypo-chi} and \ref{hypo-T}, we now define the
smooth Feshbach--Schur map with parameters $\chi$ and $T$, which
we denote by $\cF_{\chi,T}$. Its domain $\dom[\cF_{\chi,T}]$ consists
of closed operators $H$ on $\fH$ of the form
\begin{align}\label{HTW}
H \ = \ T + W_T \, , 
\end{align}
where $W_T$ is regarded a small perturbation of $T$ in the sense that
$H$ and $T$ have the same domain $\fD$ and
the corresponding graph norms are equivalent, i.e., there exist
a constant $c > 0$ such that
\begin{align}\label{HTgraphnorm}
\forall \, \phi \in \fD: 
c \big( \|T \phi\| + \|\phi| \big) 
\ \leq \ 
\|H \phi\| + \|\phi| 
\ \leq \ 
c^{-1} \big( \|T \phi\| + \|\phi| \big) \, . 
\end{align}
Furthermore, we assume 
the operator $H_{\bchi, T}: \fD_\bchi \to \fH_\bchi$ defined by
\begin{align}\label{Jpto1}
H_{\bchi, T} \ := \ T_\bchi \, + \, W_{\bchi,T} \, , 
\quad \text{with} \quad 
W_{\bchi,T} \ := \ \bchi \, W_T \, \bchi \, ,
\end{align}
is bounded invertible and that
\begin{align}
\bchi \, (H_{\bchi, T})^{-1} \, \bchi \, W_T \, \chi
\end{align}  
defines a bounded operator $\fH_\chi \to \fH_\bchi$.
 
Given a closed operator $H$ on $\fH$ possessing these
properties, the smooth Feshbach--Schur map assigns to 
$H$ the operator $F_{\chi,T}(H): \fD_\chi \to \fH_\chi$, 
\begin{align}\label{FchiT}
F_{\chi,T}(H) \ := \ 
H_{\chi,T} \, - \, 
\chi \, W_T \, \bchi \, (H_{\bchi, T})^{-1} \, \bchi \, W_T \, \chi \, ,
\end{align}
where $H_{\chi,T} : \fD_\bchi \to \fH_\bchi$ is defined by 
\begin{align}\label{Hchi}
H_{\chi,T} \ := \ T_\chi \, + \, W_{\chi,T}
\, , \quad \text{with} \quad
W_{\chi,T} \ := \ \chi \, W_T \, \chi \, . 
\end{align}
The smooth Feshbach--Schur map $F_{\chi,T}$ assigns to every
operator $H$ on $\fH$ in its domain the operator $F_{\chi,T}(H)$ on
$\fH_\chi$. Its key property is its isospectrality: $H$ is bounded
invertible (on $\fH$) if, and only if, $F_{\chi,T}(H)$ is bounded
invertible on $\fH_\chi$, their kernels have the same dimension, and
the spectral types of both operators at $0$ are the same.

Note that the choice of the auxiliary operator $T$ in $F_{\chi,T}$ is
not unique. The key observation of the present paper, however, is
that different choices of $T$ actually yield the same smooth Feshbach--Schur
map, provided they share suitable properties, which we make precise in
the following Theorem~\ref{IND}, which asserts that the smooth
Feshbach--Schur map only depends on $T P_{\chi\bchi}$ (see
Theorem~\ref{thm-ambiguity} below in Section~\ref{subsec:SFM.1} for
the proof):
%
\begin{theorem} \label{IND}
Assume Hypothesis~\ref{hypo-chi}, and let
$S$ and $T$ be two operators that fulfill
Hypothesis~\ref{hypo-T} with $\fD := \fD(S) = \fD(T)$
and such that $T \big|_{\fH_{\chi \bchi}} = S \big|_{\fH_{\chi \bchi}}$.
An operator $H: \fD \to \fH$ belongs to the domain of 
$F_{\chi,S}$ if, and only if, it belongs to the domain of 
$F_{\chi,T}$, i.e., $\dom[F_{\chi,S}] = \dom[F_{\chi,T}]$. 
In either case
\begin{align} \label{eq-domfeshST-01}
F_{\chi,S}(H) \ = \ F_{\chi,T}(H)
\end{align}   
on $\fD_\chi \subseteq \fH_\chi$.
\end{theorem}
%
Note that if $T$ fulfills Hypothesis~\ref{hypo-T} and 
if $T_\chi$ is additionally bounded, then $S := T_\bchi$ is an admissible 
choice in Theorem~\ref{IND}, and we obtain the following corollary.
%
\begin{cor} \label{cor-IND}
Assume Hypothesis~\ref{hypo-chi} for $\chi$ 
and Hypothesis~\ref{hypo-T} for $T$, and suppose that
the restriction $T_\chi$ of $T$ to $\fH_\chi$ is bounded.
Then
\begin{align} \label{eq-domfeshST-01a}
F_{\chi,T} \ = \ F_{\chi, T P_\bchi} \, .
\end{align}   
\end{cor}
%
We remark that our proof of Theorem~\ref{IND} shows 
that formally $T$ could be replaced even by $T P_{\chi \bchi}$
in the Feshbach--Schur map, i.e., $F_{\chi,T} = F_{\chi, T P_{\chi \bchi}}$.
The operator $T P_{\chi \bchi}$ does not, however, fulfill
the invertibility condition \eqref{lele-b} in Hypothesis~\ref{hypo-T},
in general.

\subsubsection{The Role of $W_{T}$ in the Smooth Feshbach--Schur Map} 
\label{Extra}
%
The term $W_T$ in the smooth Feshbach--Schur map $F_{\chi,T}$ might have more
or less convenient features depending on $T$ and the same holds true
for the operator $H_{\bchi, T}$. Convenient properties of $W_T$ do not
correspond to a better manageability of $H_{\bchi, T}$. For this
reason, an expression of $F_{\chi,T}$ in terms of $W_S$, for some
other auxiliary operator $S$ replacing $T$, is important. In this
section we address this issue.

For simplicity, in order to avoid domain issues, we assume that the
operators $H$, $T$ and $S$ are bounded, that $T$, $\chi$ and $S$
commute with one another and that $H$ belongs to the domain of
$F_{\chi,T}$.
%
\begin{definition} \label{Delta}
Suppose that $\chi$ fulfills Hypothesis~\ref{hypo-chi}, 
$T \in \cB(\fH)$ fulfills Hypothesis~\ref{hypo-T}, 
and $H \in \cB(\fH) \cap \dom[\cF_{\chi,T}]$.
For $S \in \cB(\fH)$ fulfilling Hypothesis~\ref{hypo-T}
and commuting with $T$ we define
\begin{align} \label{eq-Extra-01}
\Delta_{\chi,T}(S) \ := \ T \, \chi^2 \, + \, S \, \bchi^2 \, .
\end{align}
If $\Delta_{\chi,T}(S)\big|_{\fH_\bchi}$ is bounded invertible, we
define
\begin{align}\label{f}
f_{\chi,T}(S) \ :=  \  
\Big( T \frac{1}{\Delta_{\chi,T}(S) } \Big) \Big|_{\fH_\bchi}  
\, \oplus \, \bfone_{\fH_\bchi^\perp} \, ,   
\end{align} 
where we use the representation $\fH = \fH_\bchi \oplus \fH_\bchi^\perp$.
\end{definition}
%
The following theorem, whose proof is given in
Section~\ref{Proof-of-Tmaintsyeah}, expresses $F_{\chi,T}(H)$ in terms
of $W_S$:
%
\begin{theorem}\label{Tmaintsyeah}
For every $\chi$, $T$, $H$, and $S$ as in Definition~\ref{Delta}, it
follows that
\begin{align}\label{maintsyeah}
F_{\chi,T}(H) \ = \ &
S \, f_{\chi,T}(S) \, + \, 
\chi \, f_{\chi,T}(S) \, W_S \, f_{\chi,T}(S) \, \chi 
\nonumber \\[1ex] 
& \ \ - \, \chi \, f_{\chi,T}(S) \, W_S \, \bchi \, (H_{\bchi,T})^{-1}
\, \bchi \, W_S \, f_{\chi,T}(S) \, \chi \, ,    
\end{align}  
where $H_{\bchi,T}$ is defined in \eqref{Jpto1}. 
\end{theorem}

\subsection{The Renormalization Flow }\label{RF}

\subsubsection{Operators on Fock Spaces and their Kernels} 
\label{basicnotation}
%
The renormalization map is defined for operators acting on the boson
Fock space over a one-particle Hilbert space $\fh$, which, in this
paper, is assumed to be $\fh = L^2(\RR^3)$.  We introduce some
notation.

We denote by $\fF$ the bosonic Fock space defined by 
\begin{align}
\fF \ :=  \ \bigoplus_{n = 0}^\infty \fF_n,  
\end{align}  
endowed with the inner product of the direct sum. Here
\begin{align}
\fF_n \ :=  \ L_{\mathrm{sym}}^{2}(\RR^{3n}) \ \equiv \ 
\bigotimes_{\mathrm{sym}}^n L^2(\RR^3) \, , \quad  \fF_0 = \CC \, ,    
\end{align}  
and $L_{\mathrm{sym}}^{2}(\RR^{3n}) \subset L^{2}(\RR^{3n})$ is the
subspace of the totally symmetric functions in $L^{2}(\RR^{3n})$,
i.e., square-integrable functions $\phi_n \in L^{2}(\RR^{3n})$, which
obey $\phi_n(k_1, k_2, \cdots, k_n) = \phi_n(k_{\pi(1)}, k_{\pi(2)},
\cdots, k_{\pi(n)})$, for every permutation $\pi \in \cS_n$. We denote
the elements of $\fF$ by sequences
\begin{align}
\psi \ = \ (\psi_{n})_{n = 0}^{\infty} \, ,  \quad \psi_n \in \fF_n \, , 
\end{align}   
so that the scalar product of 
$\phi = (\phi_{n})_{n = 0}^{\infty}, \psi = (\psi_{n})_{n = 0}^{\infty} \in \fF$ 
reads
\begin{align}
\la \phi , \psi \ra 
\ = \ 
\sum_{n=0}^\infty \la \phi_n , \psi_n \ra \, , 
\end{align}   
where it is understood that the scalar product on the left side is for
vectors in $\fF$ and on the right side for vectors in $\fF_n$. We use
the symbol
\begin{align}
\fF_\fin \ \subset \ \fF
\end{align}  
for the dense subspace of sequences 
$(\psi_{n})_{n = 0}^{\infty} \in \fF$, for which all but finitely many
$\psi_n$ are zero and all non-zero $\psi_n$ are Schwartz test
functions. The free boson operator on $\fF$ is denoted by $H_\ph$,
with
\begin{align}\label{free-1}
H_\ph ( \psi_n)_{n=0}^\infty \ =: \ (\phi_n)_{n= 0}^\infty \, ,   
\end{align} 
where $\phi_0 = 0$ and 
\begin{align}\label{free-2}
\phi_n(k_1, \ldots, k_n) 
\ := \ 
\big( |k_1| + \ldots + |k_n| \big) \, \psi_n(k_1, \ldots, k_n) \, ,
\end{align} 
for all $n \in \NN$ and all $k_1, \ldots, k_n \in \RR^3$. Its domain
is the set of vectors $( \psi_n)_{n=0}^\infty \in \fF$ that yield an
element of $\fF$ in the right hand side of \eqref{free-1}.
In the present work we use pointwise creation and annihilation
operators. For fixed $k \in \RR^3$, the annihilation operator 
$a(k): \fF_\fin \to \fF_\fin$ is defined by 
$a(k) (\psi_n)_{n=0}^\infty := (\phi_n)_{n=0}^\infty$, with 
$\phi_0 =0$ and
\begin{align}
\phi_n(k_1, \ldots, k_n) 
\ = \ 
\sqrt{n+1} \, \psi_{n+1} (k, k_1, \ldots, k_n) \, ,
\end{align}    
for all $n \in \NN$ and all $k_1, \ldots, k_n \in \RR^3$.
The corresponding creation operator $a^*(k)$ is defined as a quadratic
form on $\fF_\fin$ by
\begin{align}
\la  a^*(k) \phi \, , \, \psi \ra 
\ := \ 
\la \psi \, ,  \, a(k) \phi \ra \, .  
\end{align}  
In the present manuscript, we use the notation:
\begin{align} \label{eq-not-01}
k_i^m \, := \, (k_i, \ldots, k_m) \, \in \, \RR^{d(m-i+1)} 
\, , \quad  &  
\tk_i^n \, := \, (\tk_i, \ldots, \tk_n) \, \in \, \RR^{d(n-i+1)} \, , 
\\[1ex] \label{eq-not-02}
|k_i^m| \ := \ \prod_{j=i}^m | k_j| 
\, , \qquad &
\rmd k_i^m \ := \ \prod_{j=i}^m \rmd^dk_j \, ,
\\[1ex] \label{eq-not-03}
a^*(k_i^m) \ := \ \prod_{j=i}^m a^*(k_j)
\, , \qquad &
a(\tk_i^n) \ := \ \prod_{j=i}^n a^*(\tk_j) \, .
\end{align}  
We study operators on $\fF$ that are perturbations of $H_\ph$. These
are defined in terms of creation and annihilation operators and
measurable functions of the form
\begin{align} \label{wmn-01}
v_{m,n}:  \cI \times B^m \times B^n \, \to \, \CC \, , \quad 
m, n  \in \NN_0 := \{0, 1, \ldots \} , 
\end{align}  
where $\cI := [0,1]$ is the closed unit interval from $0$ to $1$, 
$B^0 := \{0\}$, and $B := \{ k \in \RR^3: |k| < 1 \}$ is the open 
unit ball in $\RR^3$ centered at the origin. We assume that the functions  
$v_{m,n}(r; k_1^m; \tk_1^n )$ are 
continuously differentiable in $r \in \cI$, for almost every 
$(k_1^m ;\tk_1^n) \in B^m \times B^n$. More specifically,
we assume that $v_{m,n}$ is an element of the Banach space 
\begin{align} \label{wmn-02}
\cW_{m,n} \ := \ L^2\big[ B^m \times B^n \, ; \, C^1(\cI) \big] \, ,
\end{align}  
where the norm on $\cW_{m,n}$ is defined as
\begin{align} \label{eq:muNormtilde-01}
\| v_{m,n} \|_{\cW_{m,n}} \ := \ 
\bigg( \int_{B^{m+n}} 
\big\| v_{m,n}( \cdot ; k_1^m; \tk_1^n) \big\|_{C^1(\cI)}^2 \:
\frac{\rmd k_1^m \: \rmd\tk_1^n}{|k_1^m|^{3+2\mu} \: |\tk_1^n|^{3+2\mu} } 
\bigg)^{1/2} \, ,  
\end{align}  
Here $\mu > 0$ is an infrared regularization that is fixed throughout
this paper and omitted from the notation.

We use the decomposition 
\begin{align} \label{eq:muNormtilde-01,5}
C^1(\cI) \ = \ \CC \oplus \cT
\quad \text{where} \quad
\cT \ := \ \big\{ h \in C^1(\cI) \, \big| \; h(0) = 0 \big\}
\end{align}  
is the space of continuous differentiable functions vanishing at
$r=0$, which we equip with the norm 
\begin{align} \label{eq:muNormtilde-01,6}
\|f\|_{(\partial_r)} \ := \ \max_{r \in \cI} |\partial_r f(r)| 
\end{align}  
and any function $f \in C^1(\cI)$ is represented as 
$f(0) \oplus [f-f(0)]$.  We define the norm on $C^1(\cI)$ as
\begin{align} \label{eq:muNormtilde-02}
\| f \|_{C^1(\cI)}
\ := \ 
|f(0)| \: + \: \|f\|_{(\partial_r)} \, .
\end{align}  
Using $|f(r)| \leq |f(0)| + \max_{r \in \cI} |\partial_r f(r)|$, it is
easy to see that
\begin{align} \label{eq:muNormtilde-03}
\| f \|_{C^1(\cI)}
\ \leq \ 
\max_{r \in \cI} \big| f(r) \big|
\: + \: \max_{r \in \cI} \big| \partial_r f(r) \big| 
\ \leq \ 
2 \, \| f \|_{C^1(\cI)} \, ,
\end{align}  
so \eqref{eq:muNormtilde-02} is, indeed, equivalent to the standard
norm on $C^1(\cI)$.

In case that $m = 0$ or $n = 0$ in \eqref{eq:muNormtilde-01}, the
corresponding integral over $k_1^m$ or $\tk_1^n$, respectively, does
not appear. Note the special case $m=n=0$, when we have that
  $v_{0,0} \in C^1(\cI)$ with
\begin{align} \label{eq:muNormtilde-03a}
\| v_{0,0} \|_{\cW_{0,0}} 
\ = \ 
|v_{0,0}(0)| \: + \: \| v_{0,0} \|_{(\partial_r)}
\ = \ 
|v_{0,0}(0)| \: + \: \max_{r \in \cI} \big| \partial_r v_{0,0}(r) \big| \, .
\end{align}  
We denote by $\bbone_\cI :\RR_0^+ \to \RR_0^+$ the characteristic
function of the interval 
$\cI = [0, 1] \subseteq \RR_0^+ := [0, \infty)$, and
\begin{align}
P_\red \ :=  \ \bbone_\cI(H_\ph) \, . 
\end{align}  
Given a function $v_{m,n} \in \cW_{m,n}$ as above, we define its
\textit{quantization} to be the corresponding quadratic form
\begin{align} \label{Wmn-eq-01,1}
\qHH_{m,n}[ & v_{m,n}]  
\ := \ 
P_\red \int_{B^{m+n}} 
\bigg( \prod_{i=1} \frac{\rmd k_i}{|k_i|^{1/2}} \: a^*(k_i) \bigg)
\\
& \qquad \qquad 
v_{m,n}(H_\ph; \, k_1, \ldots, k_m; \, \tk_1, \ldots \tk_n) 
\bigg( \prod_{j=1} \frac{\rmd \tk_j}{|\tk_j|^{1/2}} \: a(\tk_j) \bigg) 
P_\red \, ,
\end{align}  
on $\fF_\fin$, which, using the
notation~\eqref{eq-not-01}-\eqref{eq-not-03}, reads
\begin{align} \label{Wmn-eq-01,2}
\qHH_{m,n} & [v_{m,n}]  
\ = \ 
\\ \nonumber & P_\red \bigg( \int_{B^{m+n}} a^*(k_1^m) \;
v_{m,n}(H_\ph; k_1^m; \tk_1^n) \;
a(\tk_1^n) \; 
\frac{\rmd k_1^m \: \rmd \tk_1^n}{|k_1^m|^{1/2} \: |\tk_1^n|^{1/2}} 
\bigg) P_\red \, .
\end{align}  
The quadratic form $\qHH_{m,n}[v_{m,n}]$ is represented by a bounded
operator which we also denote $\qHH_{m,n}[v_{m,n}]$, in slight abuse
of notation. This operator fulfills the norm bound
\begin{align} \label{eq-Wmn-01}
\big\| \qHH_{m,n}[v_{m,n}] \big\|_\op 
\ \leq \ 
\frac{\|v_{m,n}\|_{\cW_{m,n}} }{\sqrt{m^m n^n}} \, ,
\end{align}  
where $0^0 = 1$ and $\| \cdot \|_\op := \| \cdot \|_{\cB(\fF)}$
denotes the operator norm on $\fF$ (see 
\cite[Thms~3.1 and 3.3]{BachChenFroehlichSigal2003}).

In this paper, we consider sequences of functions of the form 
\begin{align} \label{eq-Wmn-02}
\uv \ := \ (v_{m,n})_{m + n \geq 0} \, .
\end{align}  
The component $v_{0,0}$ plays a distinguished role and is called
the \textit{free part} of $\uv$, while
\begin{align}\label{vI-01}
\uv_\gI \ = \ 
\big( v_{\gI; \, m,n} \big)_{m+n \geq 0} \, ,  
\quad 
v_{\gI; \, m,n} \ := \ 
\left\{
\begin{array}{cc}
0 & m=n=0 \, , \\
v_{m,n} & m+n \geq 1 \, ,
\end{array} 
\right.  
\end{align}
is called \textit{\uline{i}nteraction kernel} in $\uv$.
Writing $\uv_\free := (v_{0,0}, 0, 0, \ldots)$, we note the
decomposition
\begin{align}\label{vI-02}
\uv \ = \ \uv_\free \, + \, \uv_\gI \, ,
\end{align}
of $\uv$ into its free part plus its interaction.

For every $\xi \in (0,1)$, we define 
\begin{align}\label{yeahhh}
\| \uv \|^{(\xi)} 
\ := \ 
\sum_{m+n \geq 0} \xi^{-(m+n)} \| v_{m,n} \|_{\cW_{m,n}} \, ,
\end{align}  
and denote by $\cW^\xi$ the Banach space of sequences $\uv$ defined by
the norm~\eqref{yeahhh}. It follows from \eqref{eq-Wmn-01} and
\eqref{yeahhh} for every $\xi \in (0,1)$ and every $\uv \in \cW^\xi$,
that the series $\sum_{m+n \geq 0} \qHH_{m,n}[v_{m,n}]$ converges in
operator norm. We define
\begin{align} \label{HW-1}
\qHH[\uv] \ := \ & \sum_{m+n \geq 0} \qHH_{m,n}[v_{m,n}] \, . 
\\[1ex] \label{HW-2} 
\qWW[\uv] \ := \ &
\qHH[ \uv_\gI] \ = \ \sum_{m+n \geq 1} \qHH_{m,n}[v_{m,n}] \, . 
\end{align} 
It follows again from \eqref{eq-Wmn-01} and \eqref{yeahhh} that
\begin{align}\label{opbound}
\big\| \qHH[\uv] \big\|_\op \ \leq \ \|\uv\|^{(\xi)} 
\quad \text{and} \quad
\big\| \qWW[\uv] \big\|_\op \ \leq \ 
\xi \, \|\uv_\gI\|^{(\xi)} \, .
\end{align}  
The number $0 < \xi < 1$ is an expansion parameter that ensures the
summability in \eqref{HW-1} and \eqref{HW-2}. The operators of the
form $\qHH[\uv]$ are the object of study of the present paper. We are
interested in the spectral properties of these operators. We recall
that the spectrum of $\qHH[\uv]$ is the set of complex numbers $\zeta$,
called \textit{spectral parameter}, such that $\qHH[\uv] - \zeta$ is not
invertible. In this paper we restrict the spectral parameter $\zeta$
to lie in the closed disc $\bD(\tfrac{1}{4})$ of radius $\frac{1}{4}$,
centered at zero. It is convenient to include (minus) the spectral
parameter $z = -\zeta$ in $\uv$. With this in mind, we denote by
$\uw(z) = \big( w_{m,n}(z) \big)_{m+n \geq 0}$ $z$-dependent sequences
of functions such that
\begin{align} \label{wmn-03}
w_{m,n}(z) \ \in \ \cW_{m,n} \, , \quad m, n  \in \NN_0 \, , 
\end{align}  
for every $z \in D(\tfrac{1}{4})$, where $\cW_{m,n}$ is defined in
\eqref{wmn-02} and
\begin{align} \label{wmn-04}
D(r) \ := \ \{ z \in \CC \; : \ |z | < r \} \, , \quad  
\bD(r) \ := \ \{ z \in \CC \; : \ |z | \leq r \} \ \subseteq \ \CC  
\end{align}  
is the open or closed disc of radius $r>0$ centered at zero,
respectively. We assume that $\uw(z)$ satisfies the properties of
$\uv$ described in Section~\ref{basicnotation} and belongs to
$\cW^\xi$, for each $z \in \bD(\tfrac{1}{4})$. Furthermore, the map
\begin{align} \label{wmn-05}
\uw: \bD(\tfrac{1}{4}) \, \to \, \cW^\xi \, , \quad
z \, \mapsto \, \uw(z) = \big( w_{m,n}(z) \big)_{m, n \in \NN_0}
\end{align}  
is assumed to be analytic on $D(\tfrac{1}{4})$ and itself and its
complex derivative to be continuous on $\bD(\tfrac{1}{4})$ (as a map
from a subset of the complex plane with values in a complex Banach
space). We collect these maps in
\begin{align} \label{wmn-06}
\cW^\xi_Z \ := \ 
\big\{ \uw \, \in \, C^1(\bD(\tfrac{1}{4}); \cW^\xi) \; \big| \ 
\text{$\uw$ is analytic on $D(\tfrac{1}{4})$} \} \, ,
\end{align}  
which itself is a Banach space 
$\big( \cW^\xi_Z, \| \cdot \|_Z^\gI \big)$ with norm
\begin{align} \label{wmn-06,2}
\| \uw \|^{(\xi)}_{Z} 
\ := \ 
\max\Big\{ \| \uw(z) \|^{(\xi)} + \| \partial_z \uw(z) \|^{(\xi)} 
\; \Big| \ |z| \leq \tfrac{1}{4} \Big\} \, .
\end{align}  
We remark that this definition together with \eqref{opbound} implies
\begin{align} \label{wmn-06,3}
\big\| \qHH[\uw] \big\|_\op \, + \, \big\| \qHH[\partial_z \uw] \big\|_\op 
\ \leq \ 
\| \uw \|^{(\xi)}_{Z} \, . 
\end{align}  
Moreover, we identify
\begin{align}\label{identii}
w_{m,n}(z)(r; k_1^m; \tk_1^n ) 
\ \equiv \
w_{m,n}(z ; r; k_1^m; \tk_1^n ) \, ,
\end{align}  
and we observe that these functions are analytic in $z$, pointwise 
for all $r \in I$ and almost every $(k_1^m, \tk_1^n) \in B^{m+n}$.  
We frequently omit the subscript ``$Z$'' whenever no confusion
arises and identify
\begin{align}
\cW^\xi_Z  \: \equiv \: \cW^\xi 
\, , \quad  
\| \uw \|^{(\xi)}_{Z}  \: \equiv \: \| \uw \|^{(\xi)}\, , 
\end{align}  
where we use \eqref{vI-01}. Although the definition of the operators
$H[\uw(z)]$ does not require special regularity properties of the
kernel $\uw$, the renormalization analysis does.  Specifically, we
require the kernels $\uw_{m,n}$ not only to be continuous, but rather
continuously differentiable in $r \in \cI$, with a square-integrable
$C^1(\cI)$-norm, as defined in
\eqref{eq:muNormtilde-01}-\eqref{eq:muNormtilde-02}, and we
additionally assume
\begin{align}\label{w00z}
w_{0,0}(z, 0) \ = \ z \, . 
\end{align}  
We denote by 
\begin{align}\label{ur}
\uline{r+z} 
\ = \ 
\big( (r + z)_{m,n} \big)_{m+n\geq 0} 
\end{align}  
the kernel satisfying 
$P_\red (H_\ph + z) P_\red = \qHH[\uline{r + z}]$, i.e., 
$(r+z)_{0,0} = r + z$ and $(r+z)_{m,n} = 0$, whenever $m+n \geq 1$.
Note that $\uline{r+z}_\free = \uline{r+z}$ and $\uline{r+z}_\gI = 0$,
using the decomposition \eqref{vI-02}.  In general, we require
sufficiently small 
$\| (\uw - \uline{r + z})_\gI \|^{(\xi)} = \| \uw_\gI \|^{(\xi)}$. 
Further note that we could derive bounds on 
$\| \partial_z \uw(z) \|^{(\xi)}$ from $\| \uw(z) \|^{(\xi)}$, for
$|z| < \tfrac{1}{4}$, using Cauchy's estimate. It is, however,
convenient to retain the derivative $\partial_z \uw(z)$ in
Norm~\eqref{wmn-06,2} explicitly.

Finally, all operators in this paper are implicitly assumed to act on
$\Ran(P_\red)$, and we identify
\begin{align}
H_\ph + z 
\ \equiv \
P_\red \, (H_\ph + z) 
\ = \ 
\qHH\big[ \uline{r + z} \big] \, .  
\end{align}  
%

\subsubsection{Re-scaled Smooth Feshbach--Schur Map}
%
Given a number $\alpha \in \RR$, we define the scaling (or
dilation) unitary operator, $\Gamma_\alpha$, on $\fF$ by its action on
the $n$-boson sector $\fF_n$ by
\begin{align}
\big(\Gamma_\alpha \psi_n \big)(k_1, \ldots, k_n) 
\ = \ 
e^{-\frac{3 n\alpha}{2}} \psi_n(e^{-\alpha}k_1, \ldots, e^{-\alpha}k_n)
\end{align}  
and $\Gamma_\alpha \Om = \Om$. For every operator $A$ on $\fF$
we define
\begin{align}\label{Sal}
S_\alpha(A) \ := \ e^\alpha \, \Gamma_\alpha \, A \, \Gamma_\alpha^* \, ,
\end{align}  
and we call $S_\alpha$ the \textit{scaling transformation}.  

\begin{definition} \label{dfn-smoothfamily}
A collection $\{ \chi_\alpha, \bchi_\alpha \}_{\alpha > 0}$ of
smooth functions $\chi_\alpha, \bchi_\alpha: \RR_0^+ \to \cI = [0,1]$
is called a \textbf{smooth family} if 
\begin{align} \label{eq-functionchisemigroup}
\chi_{\alpha+\beta}[r] \; = \; \chi_\beta[e^\alpha r] \, \chi_\alpha[r]
\, , \quad 
\chi_\alpha^2[r] + \bchi_\alpha^2[r] \; = \; 1
\, , \quad 
\text{$\chi_\alpha \equiv 1$ on $[0, \tfrac{1}{2} e^{-\alpha}]$,}
\end{align}  
for all $r \geq 0$ and all $\alpha, \beta > 0$. 
\end{definition} 
%
Note that, if $\{\chi_\alpha, \bchi_\alpha\}_{\alpha > 0}$ is a smooth
family, then $\alpha \mapsto \chi_\alpha(r)$ is monotonically
decreasing, pointwise in $r \geq 0$.
%
\begin{bem} We show how to construct smooth families 
$\{ \chi_\alpha, \bchi_\alpha \}_{\alpha > 0}$. Given a smooth
decreasing function $\eta: \RR_0^+ \to [0,1]$ such that 
$\eta \equiv 1$ on $[0, \frac{1}{2}]$ and $\eta \equiv 0$ on 
$[1, \infty)$ and given $\alpha >0$, we define
\begin{align} \label{eq-functionchisemigroup-2}
\theta_\alpha(r) 
\ := \ 
\sin\bigg( \frac{\pi}{2} \, 
\frac{\eta(e^{\alpha}r)}{\eta(r)} \bigg)
\, , \quad
\btheta_\alpha(r) \ := \ 
\cos\bigg( \frac{\pi}{2} \, 
\frac{\eta(e^{\alpha}r)}{\eta(r)} \bigg) \, ,
\end{align}
whenever $\eta(r) > 0$ and $\theta_\alpha(r) := 0$, 
$\btheta_\alpha(r) = 1$ in case that $\eta(r)=0$. It is easy to check
that $\{ \theta_\alpha, \btheta_\alpha \}_{\alpha > 0}$ is a family of
smooth functions $\RR_0^+ \to \cI$ satisfying
\eqref{eq-functionchisemigroup}. The use of sine and cosine in
\eqref{eq-functionchisemigroup-2} ensure the smoothness of both
$\theta_\alpha$ and $\btheta_\alpha$, for $\alpha >0$.
Note that, pointwise in $r \in \cI$, we obtain
$\theta_0(r) := \lim_{\alpha \searrow 0} \theta_\alpha(r) = \bfone_{[0,m)}(r)$, 
where $m = \sup\{ r > 0| \: \eta(r) > 0 \} \in (\tfrac{1}{2}, 1]$
in the limit $\alpha \searrow 0$.
\end{bem}

Given a smooth family $\{ \chi_\alpha, \bchi_\alpha \}_{\alpha > 0}$, 
we identify
\begin{align*}
\chi_\alpha \ \equiv \ \chi_\alpha[H_\ph] 
\quad \text{and} \quad 
\bchi_\alpha \ \equiv \ \bchi_\alpha[H_\ph] 
\end{align*}  
in the following.
%
\begin{definition}\label{def:rescaledfsmap}
Let $\{ \chi_\alpha, \bchi_\alpha \}_{\alpha > 0}$ be a smooth family
in the sense of Definition~\ref{dfn-smoothfamily}. 
For every $\uw \in \cW_Z^\xi$,  
such that $\qHH[\uw(z)]$ belongs to the domain of 
$F_{\chi_\alpha, H_\ph+z}$, for all $z \in \bD(\tfrac{1}{4})$, 
we define 
\begin{align} \label{eq-functionchisemigroup-3}
\hcR_\alpha \big( \qHH[\uw(z)] \big) 
\ := \  
S_\alpha\Big( F_{\chi_\alpha, H_\ph+z} \big( \qHH[\uw(z)] \big) \Big) \, . 
\end{align}
In this case, we say that $\qHH[\uw(z)] $ belongs to the domain of 
$\hcR_\alpha$. We call $\alpha > 0$ the \textit{scaling parameter}.  
\end{definition}

\subsubsection{Main Theorems: \\ Renormalization of the Spectral Parameter}
%

\begin{definition} \label{def-Qgamma}
Let $\{ \chi_\alpha, \bchi_\alpha \}_{\alpha > 0}$ be a smooth family 
in the sense of Definition~\ref{dfn-smoothfamily}.
For every $\uw \in \cW_Z^\xi$, such that $\qHH[\uw(z)]$ belongs to the 
domain of $F_{\chi_\alpha, H_\ph+z}$, for all $z \in \bD(\tfrac{1}{4})$, 
we define 
\begin{align} \label{Qgamma}
Q_\alpha(z) \ := \ 
\Big\la \hcR_\alpha\big( \qHH[\uw(z)] \big) \Big\ra_\Om \, ,
\end{align}  
where $\la A \ra_{\Om} = \la \Om | A \Om \ra$ is the vacuum
expectation value.
\end{definition}
%
The next result is proved in Theorem~\ref{thm-specrg} below  
%
\begin{prop}\label{prop-specrg}
Let $\alpha > 0$ and $0 < r_Z < \tfrac{1}{4}$. Suppose that 
$\uw \in \cW_Z^\xi$ and that $\uw$ fulfills \eqref{w00z}. Then, for
sufficiently small $\|\uw_\gI \|^{(\xi)}$ and 
$\| \uw_\free - r \|_{(\partial_r)}$, the map $Q_\alpha$ is a
biholomorphic function from 
$D(\tfrac{1}{4} e^{-\alpha}) \cap Q_\alpha^{-1}\big( D(r_Z) \big)$
onto $D(r_Z)$.
\end{prop}

\begin{definition} \label{E}
Suppose that 
$Q_\alpha: D(\tfrac{1}{4} e^{- \alpha}) \cap Q_\alpha^{-1}\big( D(r_Z) \big) 
\to D(r_Z)$ defines a biholomorphic function as in 
Proposition~\ref{prop-specrg}. We denote by 
\begin{align} \label{eq-boEalpha}
\boE_\alpha \ \equiv \ \boE_{\alpha, \uw}:  \ 
D(r_Z) \ \to \ D(\tfrac{1}{4} e^{-\alpha}) \cap Q_\alpha^{-1}\big( D(r_Z) \big) 
\end{align}  
the inverse of $Q_\alpha$. The function $\boE_\alpha $ is called 
\textbf{renormalization of the spectral parameter}. 
\end{definition}
%
The image of $\boE_\alpha$ localizes the spectral points 
$\zeta = \boE_\alpha(z)$, where the smooth Feshbach--Schur map (or,
equivalently, the smooth rescaled Feshbach--Schur map) is well-defined,
i.e.\ those 
$\zeta \in D(\tfrac{1}{4} e^{-\alpha}) \cap Q_\alpha^{-1}\big(D(r_Z)\big)$,
for which $\qHH[\uw(\zeta)]$ is in the domain of $\hcR_\alpha$.
%
\begin{theorem}[Main Theorem: Spectral Parameter] \label{renspect}
Let $0 < r_Z < \tfrac{1}{4}$. Suppose that $\alpha > 0$ and that $\uw$
satisfies the hypothesis of Proposition~\ref{prop-specrg}, for $r_Z$
and $\alpha$. Suppose further that $\beta > 0$ and that there is 
$\tuw \in \cW^{(\xi)}$ satisfying the hypothesis of
Proposition~\ref{prop-specrg}, for $r_Z$ and $\beta$,
such that $\qHH[\tuw(z)] = \cR_\alpha(\qHH[\uw(z)])$. Then
\begin{align} \label{eq-cocyclicity}
\boE_{\alpha + \beta, \uw} 
\ = \ 
\boE_{\alpha, \uw} \circ \boE_{\beta, \uline{\widetilde{w}}} \, . 
\end{align}  
\end{theorem}

\subsubsection{Main Theorems:  Renormalization  Flow}
%
\begin{definition}[Renormalization Operator]
Assuming the requirements of Definition~\ref{E}, we define 
the renormalization map $\cR_\alpha$ as 
\begin{align} \label{eq-rgmap}
\cR_\alpha\big( \qHH[\uw] \big)(z) 
\ := \ 
\hcR_\alpha\big( \qHH[\uw( \boE_\alpha(z))] \big) \, .  
\end{align}  
\end{definition}
%
The renormalization $\boE_\alpha$ of the spectral parameter is the key
ingredient for the flow property of the renormalization operator based
on the sharp Feshbach--Schur map, see
\cite{BachBallesterosFroehlich2015}. The choice of $\boE_\alpha$ is
determined by the requirement
\begin{align} \label{eq-vacexpz}
\big\la \cR_\alpha(\qHH[\uw])(z) \big\ra_\Om \ = \ z \, ,
\end{align}  
which is a consequence of our definitions, since     
\begin{align} \label{eq-vacexpz-2}
\big\la \cR_\gamma(\qHH[\uw])(z) \big\ra_\Om 
\ = \ 
\big\la \hcR_\gamma(\qHH[\uw(Q_\gamma^{-1}(z))]) \big\ra_\Om 
\ = \ 
Q_\gamma(Q_\gamma^{-1}(z)) \ = \ z \, .
\end{align}  
For the smooth version of the Feshbach--Schur map, the new key
ingredient for the flow property is to fix the operator $T$ in the 
smooth Feshbach--Schur map to be $H_\ph$.
%
\begin{theorem}[Main Theorem: Flow Property]\label{thm:fullsemigroup}
Assume that the hypothesis of Theorem~\ref{renspect} is fulfilled. 
Then 
\begin{align}
\forall\, z\in D_{r_Z}: \quad 
\cR_\beta \circ \cR_\alpha(\qHH[\uw(z)]) 
\ = \ 
\cR_{\alpha+\beta}(\qHH[\uw(z)]) \, .
\end{align}  
\end{theorem}

\subsection{Iterated Applications of the Smooth 
Feshbach--Schur Map}\label{Itera}
%
Theorems~\ref{renspect} and \ref{thm:fullsemigroup} require the
existence of an interaction kernel $\tuw$ such that 
$\qHH[\tuw(z)] = \cR_\alpha(\qHH[\uw(z)])$. More restrictive assumptions on
$\uw$ ensure this. Once this is achieved, the renormalization map can
be iterated any finite number of times, and the spectral analysis
described at the end of Section~\ref{intro} can be performed.

The goal of the present paper is to introduce a renormalization map
based on the smooth Feshbach--Schur map that satisfies a flow property. The
spectral analysis of operators in quantum field theory is not the
objective of this paper, and we do not include it here because this
would shift the focus away from the flow property to a series of
technically involved proofs which would lead to a considerable
increase in length and made this article difficult to read. The
iterative application method for the renormalization map was first
introduced in \cite{BachFroehlichSigal1998a, BachFroehlichSigal1998b}
for the sharp Feshbach--Schur map, and in
\cite{BachChenFroehlichSigal2003} for the smooth Feshbach--Schur map. The
corresponding results for a renormalization map based on the sharp
Feshbach--Schur map satisfying a flow property were derived in
\cite{BachBallesterosFroehlich2015}. Our results for the smooth
counterpart are presented below as an announcement, but our proofs are
deferred to a forthcoming paper.

The idea of the construction of the iterative applications of the
smooth Feshbach--Schur map is that the operators we use must be
all the time close to the free photon energy operator
$H_\ph + z = H(\uline{r+z})$, notably also after the application of
the renormalization map. This is indeed the situation we consider in
this paper, since the central requirement of our main theorems is that
$\| (\uw - \uline{r + z})_\gI \|^{(\xi)} = \| (\uw)_\gI \|^{(\xi)}$ and
$\| (\uw - \uline{r+z})_\free \|_{(\partial_r)} = 
\max_{r \in \cI} |\partial_r w_\free - 1|$ are sufficiently small.

For iterated applications of the renormalization map, we need more
refined norms. For a kernel $\uw \in \cW^\xi$, we define
\begin{align} \label{eq-newnormpartr}
\| \uw \| 
\ := \ 
\| \uw \|^{(\xi)} \: + \: \| \partial_z \uw \|^{(\xi)} 
\: + \: \| \partial_r \uw \|^{(\xi)}    
\end{align}  
and 
\begin{align} \label{normpartr}
\| \uw \|_{(\partial_z)} 
\ := \ 
\sup\big\{ |\partial_z w_{0,0}(z,r)| \: : \ 
r \in \cI \, , \ z \in D(\tfrac{1}{4}) \big\} \, . 
\end{align}  
For every 
\begin{align}
\va \ = \ (a_I, a_R, a_Z) \ \in \ (\RR^+)^3 \, , 
\end{align} 
we denote by $\cW^\xi_{\va}$ the polydisc of all
$\uw \in \cW^\xi $ such that
\begin{align}
\big\| (\uw -\uline{r, z})_I \big\| \; \leq \; a_I 
\, , \quad
\| \uw -\uline{r, z} \|_{(\partial_r)} \; \leq \; a_R
\, , \quad
\| \uw -\uline{r, z} \|_{(\partial_z)} \; \leq \; a_Z \, .
\end{align}  
For $\vb \in (\RR^+)^3$, we say that $\vb \leq \va$, if this relation
holds componentwise, i.e., $b_I \leq a_I$, $b_R \leq a_R$, and 
$b_Z \leq a_Z$. Moreover, we denote $|\vb| = b_I + b_R + b_Z$. The
existence of iterative applications of the renormalization map is a
consequence of the following theorem, which states that, for a
suitably chosen polydisc $\cW^\xi_{\va}$ and sufficiently small
initial data, their orbits under iterated applications of the
renormalization map never leave this polydisc.
%
\begin{announcement} \label{ann}
There exist $\boldsymbol{\va} \in (\RR^+)^3 $, $\xi > 0$, and a closed, 
non-empty interval $\boI \neq \emptyset$ such that, for every 
$\uw \in \cW^\xi_{\boldsymbol{\va}}$ and every $\alpha \in \boldsymbol{I}$, 
$\uw$ satisfies all norm-bounds of the hypotheses of  
Theorems~\ref{renspect} and \ref{thm:fullsemigroup}. 

Moreover, for $\vec \epsilon \leq \vec {\boldsymbol{a}}$ with 
$|\vec \epsilon|$ sufficiently small, $H(\uw)$ belongs to the domain 
of $( R_\alpha )^{n}$, for all $n \in \NN$, and there exists 
$\uw^{(n)} \in  \cW^\xi_{\vec{\boldsymbol{\alpha}}}$ such that
\begin{align}\label{Rn}
\qHH[\uw^{(n)}] \ = \ R_\alpha^{n} \big( \qHH[\uw] \big).   
\end{align}  
The interaction part of $\uw^{(n)}$ contracts exponentially, as $n$
tends to infinity, i.e., there exists $\ell < 1$ such  that
\begin{align}\label{contractionn}
\big\| (\uw^{(n)})_\gI \big\|_Z^{(\xi)}
\ \leq \ \ell^n \, .  
\end{align}  
\end{announcement}

The iteration of the renormalization map in Theorem~\ref{ann} is
possible by the contraction property~\eqref{contractionn}, because it
implies that the kernel $\uw^{(n)} $ gets ever closer (exponentially
fast) to the free kernel, as $n$ tends to infinity.  The fact,
however, that we take $H_\ph$ as a parameter in the smooth
Feshbach--Schur map in Definition~\ref{def:rescaledfsmap} (and
$W(\uw^{(n)}) \neq H_\ph $) causes some problems when controlling the
free part $w^{(n)}_{0,0}$ (here Theorem~\ref{Tmaintsyeah} is useful).
The key observation to solve these problems is that the norm estimates
for $w^{(n)}_{0,0}(z,r)$ depend on $w^{(n-1)}_{0,0}(z,r)$ only for $r
\in [0, \tfrac{1}{2}]$.  More specifically, defining
\begin{align}\label{normpartialr}
\| \uw \|_{(\partial_r, -)} 
\ := \ &
\sup\big\{ |\partial_r w_{0,0}(z,r)| \: : \ 
r \in [0, \tfrac{1}{2}] \, , \ z \in D(\tfrac{1}{4}) \big\} \, , 
\\[1ex] \label{normpartialz}
\| \uw \|_{(\partial_z, -)} 
\ := \ &
\sup\big\{ |\partial_z w_{0,0}(z,r)| \: : \ 
r \in [0, \tfrac{1}{2}] \, , \ z \in D(\tfrac{1}{4}) \big\} \, , 
\end{align}  
we make use of the fact that the norms 
$\| \uw^{(n)} - \uline{r+z} \|_{(\partial_r)}$ and 
$\| \uw^{(n)} - \uline{r+z} \|_{(\partial_z)}$ 
can be estimated solely in terms of 
$\| \uw^{(n-1)} - \uline{r+z} \|_{(\partial_r, -)}$ and 
$\| \uw^{(n-1)} - \uline{r+z} \|_{(\partial_z, -)}$.

\section{Proofs of Section~\ref{SFSM}
The Smooth Feshbach--Schur Map}\label{sec:SFM}
%
\subsection{Proof of Theorem~\ref{IND}} \label{subsec:SFM.1}
%
In this section we study the smooth Feshbach--Schur map introduced in
Section~\ref{SFSM} in detail and prove Theorem~\ref{IND}. We use the
notation introduced in the latter section.
%
\begin{prop} \label{prop-isospecSFB}
Suppose that $H: \fD \to \fH$ belongs to the domain of 
$F_{\chi,T}$ (see Section~\ref{SFSM}). Then $H$ is bounded 
invertible if and only if $F_{\chi,T}(H): \fD_\chi \to \fH_\chi$ 
is bounded invertible [recall \eqref{hilbert-1}]. In either case:      
\begin{align} \label{eq-Fsmoothinvers}
F_{\chi,T}(H)^{-1} \ = \  
\Big[ \chi H^{-1}\chi + \bchi T^{-1}\bchi  \Big]\Big|_{\fH_\chi}.
\end{align}  
Moreover, $\dim[\cern(H)] = \dim[\cern(F_{\chi,T}(H))]$.
\begin{proof}
See \cite[Theorem~1]{GriesemerHasler2008} and
\cite[Theorem~2.1]{BachChenFroehlichSigal2003}.
\end{proof}
\end{prop}
%
We refer to  Proposition~\ref{prop-isospecSFB} 
as \uline{\textit{isospectrality of $H$ and $F_{\chi,T}(H)$}}. 
Moreover, the operators $\chi$ and $\bchi$ are referred to as 
\uline{\textit{smooth projections}}.
%
\begin{bem} 
In case that $\chi = \chi^2$ is an orthogonal projection,  
we have that $P_\chi = \chi =: P$, $\bchi = P_\bchi = P^\perp$,
and $\fH_{P^\perp} = \fH_P^\perp$. In this case, Eq.~\eqref{Jpto1} 
becomes 
\begin{align} \label{eq-projfesh-01}
H_{P^\perp,T} 
\ = \  
T_{P^\perp} + P^\perp W_T P^\perp
\ = \ 
\big( P^{\perp} H P^{\perp} \big) \Big|_{\fH_P^\perp} \, ,   
\end{align}  
independently of $T$, and we write 
$H_{P^\perp} := H_{P^\perp,T}$. Similarly, for $H$ in the domain of
$F_{\chi,T}$ it follows from Eqs~\eqref{HTW} and \eqref{FchiT} that
\begin{align} \label{eq-projfesh-02}
F_{P, T}(H) 
\ = \ 
\big(  PHP + P W P^\perp 
\big( P^\perp H P^\perp \big |_{\fH_P^\perp} \big)^{-1}  W P \big)\Big|_{\fH_P} 
\ = \  
F_P(H) \, , 
\end{align}  
independently of $T$, where $F_P$ is the (projection based)
Feshbach--Schur map introduced in \cite{BachFroehlichSigal1998b}.
\end{bem}
%
Now we prove Theorem~\ref{IND}, which shows that $T P_{\chi\bchi}$ is
the only relevant part of $T$ for the smooth Feshbach--Schur map. We restate
Theorem~\ref{IND} for the convenience of the reader.
%
\begin{theorem}\label{thm-ambiguity}
Assume Hypothesis~\ref{hypo-chi}, and let
$S$ and $T$ be two operators that fulfill
Hypothesis~\ref{hypo-T} with $\fD := \fD(S) = \fD(T)$
and such that $T|_{\fH_{\chi \bchi}} = S|_{\fH_{\chi \bchi}}$.
An operator $H: \fD \to \fH$ belongs to the domain of 
$F_{\chi,S}$ if, and only if, it belongs to the domain of 
$F_{\chi,T}$. In either case
\begin{align} \label{eq-domfeshST-02}
F_{\chi,S}(H) \ = \ F_{\chi,T}(H)
\end{align}   
on $\fD_\chi \subseteq \fH_\chi$.
\begin{proof}
We first remark that, formally, Identity~\eqref{eq-domfeshST-02} 
follows from \eqref{eq-Fsmoothinvers} because 
$T|_{\fH_{\chi \bchi}} = S|_{\fH_{\chi \bchi}}$ implies that
$\bchi T^{-1} \bchi|_{\fH_\chi} = \bchi S^{-1} \bchi|_{\fH_\chi}$
and thus
\begin{align} \label{eq-Fsmoothinvers-2}
F_{\chi,T}(H)^{-1} 
\ = \ & 
\Big[ \chi H^{-1}\chi + \bchi T^{-1}\bchi  \Big]\Big|_{\fH_\chi}
\ = \  
\Big[ \chi H^{-1}\chi + \bchi S^{-1}\bchi  \Big]\Big|_{\fH_\chi}
\nonumber \\[1ex]
\ = \ & 
F_{\chi,S}(H)^{-1} \, .
\end{align}  
This argument assumes, however, (a) that $H$ belongs to the domain of
$F_{\chi,T}$ if, and only if, it belongs to the domain of $F_{\chi,S}$
and (b) that $F_{\chi,T}(H)$ or $F_{\chi,S}(H)$ (and hence both) are
invertible.
 
For a more careful argument, we first establish (a).  Namely, we now
assume that $H$ is in the domain of $F_{\chi,T}$ and demonstrate that
$H$ also belongs to the domain of $F_{\chi,S}$.  Since $H$ belongs to
the domain of $F_{\chi,T}$ we have that $\fD(T) = \fD(H)$, and by
assumption $\fD(S) = \fD(T)$, which implies that $\fD(S) = \fD(H)$.
Next, we recall the notation of Section~\ref{SFSM},
\begin{align} \label{eq-3.03}
T + W_T \ = \ H \ = \ S + W_S \, .
\end{align}  
Multiplying by $\chi$ from the left and by $\bchi$ from the right
and using that $T$ and $S$ commute with both $\chi$ and $\bchi$,
we obtain 
\begin{align} \label{eq-3.04}
\chi W_T \bchi - \chi W_S \bchi 
\ = \ 
(S - T) \, \chi \, \bchi 
\ = \ 0 \, , 
\end{align}  
since $S P_{\chi\bchi} = T P_{\chi\bchi}$. Similarly
$\bchi W_T \chi = \bchi W_S \chi$ on $\fD$, so
\begin{align} \label{eq-3.05}
\chi W_T \bchi \ = \ \chi W_S \bchi 
\quad \text{and} \quad 
\bchi W_T \chi \ = \ \bchi W_S \chi 
\end{align}  
on $\fD$. Next we use \eqref{eq-3.04} on $\fD_\bchi$ and
$\chi^2 + \bchi^2 = 1$ to obtain
\begin{align} \label{eq-3.06}
H_{\bchi,S} \ = \ &
S_\bchi + \bchi W_S \bchi 
\ = \ 
S_\bchi + (T-S) \bchi^2 + \bchi W_T \bchi 
\\[1ex] \nonumber 
\ = \ &
T_\bchi \bchi^2 + S_\bchi \chi^2 + \bchi W_T \bchi 
\ = \ 
T_\bchi + \bchi W_T \bchi 
\ = \ 
H_{\bchi,T} \, ,
\end{align}  
because $S_\bchi \chi^2 = S P_{\chi\bchi} \chi^2 = T P_{\chi\bchi}
\chi^2 = T_\bchi \chi^2$ by assumption. Therefore, since $H_{\bchi,T}$
is bounded invertible, so is $H_{\bchi,S}$, and since
$\bchi (H_{\bchi,T})^{-1} \bchi W_T \chi$ defines a bounded
operator $\fH_\chi \to \fH_\bchi$, so does 
$\bchi (H_{\bchi,S})^{-1} \bchi W_S \chi$. This proves that $H$
belongs to the domain of $F_{\chi,S}$.

Finally, an argument similar to \eqref{eq-3.06},
interchanging the roles of $\chi$ and $\bchi$, shows that
$H_{\chi,S} = H_{\chi,T}$ on $\fD_\chi$. This together with \eqref{eq-3.06}
and \eqref{eq-3.05} implies that $F_{\chi,S}(H) = F_{\chi,T}(H)$ 
on $\fD_\chi$.
\end{proof}
\end{theorem}

\subsection{Proof of Theorem~\ref{Tmaintsyeah}} 
\label{Proof-of-Tmaintsyeah}
%
In this section we prove Theorem~\ref{Tmaintsyeah}. We use the
notation of Section~\ref{Extra} and assume the corresponding
hypotheses. We recall from \eqref{hilbert-2,1} that 
$\bchi = \bchi P_\bchi = P_\bchi \bchi$ and from
Definition~\ref{Delta} that $\Delta_{\chi,T}(S) = T\chi^2 + S\bchi^2$.
In the following remark we motivate the definition of $f_{\chi,T}(S)$
in Eq.~\eqref{f}. For this we recall from \eqref{hilbert-2,1} that
\begin{align} \label{ca1}
\bchi \ = \ \bchi \, P_\bchi \ = \ P_\bchi \, \bchi \, ,
\end{align}  
due to the self-adjointness of $\bchi$.

\begin{bem}\label{DelraRem}
Assuming that $\Delta_{\chi,T}(S) \big|_{\fH_\bchi}$ is bounded invertible,
the operator $f_{\chi,T}$ permits us to trade $W_T$ for $W_S$ in the
smooth Feshbach--Schur operator $F_{\chi,T}(H)$. More precisely,
we exchange in the quadratic term with respect to $W_T$ of the smooth
Feshbach--Schur map the factor $\bchi W_T$ by 
$\bchi W_S f_{\chi,T}(S)$ using that
\begin{align}\label{ca1,1}
\bchi \, W_T \: - \: \bchi \, W_S \, f_{\chi,T}(S) 
\ = \ 
H_{\bchi, T} \,  (S-T) \, \Delta_{\chi,T}(S)^{-1} \, \bchi
\end{align}  
and the factor $W_T \bchi$ by $f_{\chi,T}(S) W_S \bchi$ using that
\begin{align} \label{ca2}
W_T \, \bchi \: - \: f_{\chi,T}(S) \, W_S \, \bchi
\ = \ 
(S-T) \, \Delta_{\chi,T}(S)^{-1} \, \bchi \, H_{\bchi, T} \, . 
\end{align}  
The factor $H_{\bchi, T}$ in Eqs.~\eqref{ca1} and \eqref{ca2} cancels with  
$( H_{\bchi, T} )^{-1}$ in the smooth Feshbach--Schur map, 
see \eqref{FchiT}). This is specifically presented in Lemma~\ref{Si} below. 
\begin{proof}[Proof of Eqs.~\eqref{ca1} and \eqref{ca2}] 
We only prove Eq.~\eqref{ca1}, Eq.~\eqref{ca2} is deduced
similarly. We first observe that
\begin{align}\label{JN1}
\big( 1 - f_{\chi,T}(S) \big) \, (1 - P_\bchi) \ = \ 0 \, ,
\end{align}  
because, by definition, $f_{\chi,T}(S)$ equals the identity on 
$\fH_\bchi^\perp$. Eq.~\eqref{JN1}, the fact that $\chi^2 + \bchi^2 = 1$ 
and $\bchi = P_\bchi \bchi = \bchi P_\bchi$ [see \eqref{ca1,1}] imply that 
\begin{align} \label{Si1}
\big( 1 - f_{\chi,T}(S) \big) 
\ = \ &
\big( 1 - f_{\chi,T}(S) \big) \, P_\bchi    
\ = \  
\big( T \chi^2 + S \bchi^2  - T \big) \, \Delta_{\chi,T}(S)^{-1} \, P_\bchi
\nonumber \\[1ex] 
\ = \ & 
(S - T) \, \bchi^2 \, \Delta_{\chi,T}(S)^{-1} \, P_\bchi 
\ = \  
\bchi \, (S - T) \, \Delta_{\chi,T}(S)^{-1} \, \bchi \, .
\end{align}  
We recall that, by definition, $H = T + W_T = S+ W_S$.  
This implies that $W_T = W_S + S - T$. We use \eqref{Si1} to obtain
\begin{align}
\bchi \, W_T \, - \, \bchi \, W_S \, f_{\chi,T}(S)
\ = \ &
\bchi \, W_S \, \bchi \, (S - T) \, \Delta_{\chi,T}(S)^{-1} \, \bchi
\: + \: (S-T) \, \bchi
\nonumber \\[1ex]
\ = \ &
\big( \bchi \, W_S \, \bchi \, + \, T \, \chi^2 \, + \, S \, \bchi^2 \big)
\, (S-T) \, \Delta_{\chi,T}(S)^{-1} \, \bchi 
\nonumber \\[1ex]
\ = \ &
\big( \bchi (W_S + S) \bchi \, + \, T \, \chi^2 \big)
\, (S-T) \, \Delta_{\chi,T}(S)^{-1} \, \bchi 
\nonumber \\[1ex]
\ = \ &
H_{\bchi,T} \, (S-T) \, \Delta_{\chi,T}(S)^{-1} \, \bchi \, ,
\end{align}  
where we utilized \eqref{Jpto1}. 
\end{proof}
\end{bem}

\begin{lemma}\label{Si}
The quadratic term in $W_T$ in the smooth Feshbach--Schur
map~\eqref{FchiT} satisfies the following identity:
\begin{align}\label{long}
\chi \, W_T \, & \bchi \, H_{\bchi, T}^{-1} \, \bchi \, W_T \, \chi
\: - \: 
\chi \, f_{\chi,T}(S) \, W_S \, \bchi \, H_{\bchi, T}^{-1} \, 
\bchi \, W_S \, f_{\chi,T}(S) \, \chi
\nonumber \\[1ex] 
& \ = \    
\chi \, f_{\chi,T}(S) \, W_S \big( 1 - f_{\chi T}(S) \big) \chi
\: + \: \chi \big( 1 - f_{\chi T}(S) \big) W_T \, \chi \, .  
\end{align}  
\begin{proof}
We temporarily denote by $A$ the first line (left side) in
Eq.~\eqref{long}. A telescopic sum argument leads us to
\begin{align} \label{ss1}
A \ = \ &   
\chi \, f_{\chi,T}(S) \, W_S \, \bchi \, H_{\bchi, T}^{-1} 
\big( \bchi \, W_T \, \chi \, - \, \bchi \, W_S \, f_{\chi,T}(S) \, \chi \big)
\\[1ex] \label{ss2}  
& + \big( \chi \, W_T \, \bchi 
\, - \, \chi \, f_{\chi,T}(S) \, W_S \, \bchi \big) 
H_{\bchi, T}^{-1} \, \bchi \, W_T \, \chi \, .
\end{align}  
Using \eqref{ca1} and \eqref{Si1}, the right side in Line~\eqref{ss1}
is seen to equal
\begin{align}\label{sii4}
\chi \, f_{\chi,T}(S) \, & W_S \, \bchi \, H_\bchi^{-1} \, H_\bchi 
\, (S-T) \, \Delta_{\chi,T}(S)^{-1} \, \bchi \, \chi  
\nonumber \\[1ex] 
\ = \ &
\chi \, f_{\chi,T}(S) \, W_S \big( 1 - f_{\chi,  T}(S) \big) \chi \, ,
\end{align}  
whereas Eq.~\eqref{ss2} equals 
\begin{align}\label{sii3}
\chi \, (S-T) \, \Delta_{\chi,T}(S)^{-1} \, \bchi \, H_\bchi \, H_\bchi^{-1} 
\, \bchi \, W_T \, \chi 
\ = \ 
\chi \big( 1 - f_{\chi ,  T}(S) \big) W_T \, \chi \, ,
\end{align}  
thanks to \eqref{ca2} and \eqref{Si1}. Eqs.~\eqref{ss1}, \eqref{sii3},
and \eqref{sii4} imply the identity sought for.
\end{proof}
\end{lemma}
%
Lemma~\ref{Si} allows us to prove the main theorem of this section,
Theorem~\ref{Tmaintsyeah}, which we restate here for the convenience
of the reader.
%
\begin{theorem}
For every $H, T$ and $S$ as in Definition~\ref{Delta}, 
it follows that
\begin{align} \label{maints}
F_{\chi,T}(H) 
\ = \ & 
S \, f_{\chi,T}(S) 
\: + \: \chi \, f_{\chi,T}(S) \, W_S \, f_{\chi,T}(S) \, \chi 
\nonumber \\[1ex]
& \: - \: \chi \, f_{\chi,T}(S) \, W_S \, \bchi \, H_\bchi^{-1}
\, \bchi \, W_S \, f_{\chi,T}(S) \, \chi \, . 
\end{align}    
\begin{proof}
The definition of the smooth Feshbach--Schur map 
in \eqref{FchiT} and Lemma~\ref{Si} imply that 
\begin{align} \label{chi}
F_{\chi,T}(H) 
\ = \ 
H_{\chi,T} \: - \: &  \chi \, f_{\chi,T}(S) \, W_S 
\big( 1 - f_{\chi,T}(S) \big) \chi \: - \: 
\chi \big( 1 - f_{\chi,T}(S) \big) W_T \, \chi
\nonumber \\[1ex]  
& \: - \:  \chi \, f_{\chi,T}(S) \, W_S \, \bchi 
( H_\bchi )^{-1} \bchi \, W_S \, f_{\chi,T}(S) \, \chi \, .
\end{align}  
We recall the definition \eqref{Hchi} of $ H_{\chi,T}$ and calculate
the first line on the right side of \eqref{chi} as follows,
\begin{align}\label{chi1}
T_\chi \, + \, & \chi \, W_T \, \chi 
\, - \, \chi \, f_{\chi,T}(S) \, W_S \big( 1 - f_{\chi,T}(S) \big) \chi 
\, - \, \chi \big( 1 - f_{\chi,T}(S) \big) W_T \, \chi
\nonumber \\[1ex]  
\ = \ & 
\chi \, f_{\chi,T}(S) \, W_S \, f_{\chi,T}(S) \, \chi \, + \, T_\chi
\, - \, \chi \, f_{\chi,T}(S) \, W_S \, \chi 
\, + \, \chi \, f_{\chi,T}(S) \, W_T \, \chi \, . 
\end{align} 
Next we further calculate the last three terms in \eqref{chi1}, using
that $H = T + W_T = S+ W_S$ and therefore $W_T - W_S = S-T$. We obtain
\begin{align}\label{chi3}
T_\chi \, - \, \chi \, f_{\chi,T}(S) \, W_S \, \chi
\, + \, \chi \, f_{\chi,T}(S) \, W_T \, \chi 
\ = \  
T_\chi \, + \, \chi \, f_{\chi,T}(S) \, 
(S - T) \, \chi \, .   
\end{align} 
Recall that the operator on the right side of \eqref{chi3} acts
on $\fH_\chi$. We study its action on $\fH_\chi \cap \fH_\bchi$ and on 
$\fH_\chi \cap \fH_\bchi^\perp$ separately. We first multiply by $P_\bchi$ 
from the right and obtain
\begin{align}\label{chi4}
\big( T_\chi 
\, + \, \chi \, & f_{\chi,T}(S) \, (S-T) \, \chi \big) P_\bchi 
\ = \  
\Big( T_\chi \, + \, \frac{T}{ T \, \chi^2 + S \, \bchi^2} \, 
(S-T) \, \chi^2 \Big) P_\bchi 
\nonumber \\[1ex] 
\ = \ & 
\Big( T \, P_\chi ( T \, \chi^2 + S \, \bchi^2)
\, + \, T \, P_\chi \, (S-T) \, \chi^2 \Big)
\frac{1}{ T \, \chi^2 + S \, \bchi^2} \, P_\bchi
\nonumber \\[1ex] 
\ = \ & 
S \, P_\chi \, f_{\chi,T}(S) \, P_\bchi  
\ = \  
S_\chi \, f_{\chi,T}(S) \, P_\bchi \, , 
\end{align}  
where we use that $\chi^2 + \bchi^2 = 1$ and the definition of 
$f_{\chi,T}(S)$ in \eqref{f}. Moreover, using again \eqref{f}
and $\bchi P_\bchi^\perp = 0$, we get that 
$f_{\chi,T}(S) \, P_\bchi^\perp = P_\bchi^\perp$. Additionally, since 
$\bchi$ is self-adjoint and therefore 
$\fH_\bchi^\perp = \cern(\bchi)$, it follows that 
$\chi^2 P_\bchi^\perp = P_\bchi^\perp$, where we recall that 
$\chi^2 + \bchi^2 = 1$. This implies that 
\begin{align}\label{chi5}
\big( T_\chi \, + \, \chi \, f_{\chi,T}(S) \, (S-T) \, \chi \big) 
\, P_\bchi^\perp
\ = \ 
S_\chi \, P_\bchi^\perp  
\ = \ S_\chi \, f_{\chi,T}(S) \, P_\bchi^\perp \, .
\end{align}  
Using \eqref{chi3}, \eqref{chi4} and \eqref{chi5} we arrive at
\begin{align}\label{chi6}
T_\chi \, - \, \chi \, f_{\chi,T}(S) \, W_S \, \chi
\, +  \, \chi \, f_{\chi,T}(S) \, W_T \, \chi 
\ = \  
S_\chi \, f_{\chi,T}(S) \, .
\end{align}  
Finally, \eqref{maints} follows from \eqref{chi},
\eqref{chi1} and \eqref{chi6}. 
\end{proof}
\end{theorem}

\subsection{Isospectrality of the Smooth Feshbach--Schur Map from 
Isospectrality of the sharp Feshbach--Schur Map}
\label{subsec:sharpsmooth}
%
In this section we derive the isospectrality formulated in
Proposition~\ref{prop-isospecSFB} of the smooth Feshbach Map on $\fH$
from the isospectrality established in, e.g.,
\cite{BachFroehlichSigal1998b} of the sharp Feshbach--Schur Map
realized on the bigger Hilbert space
\begin{align} \label{eq-sharpsmooth-02}
\hfH \ := \ \fH_\chi \oplus \fH_\bchi \, ,
\end{align}
where we recall from \eqref{hilbert-1} and \eqref{hilbert-2}
that $\fH_\chi = P_\chi \fH$ and $\fH_\bchi = P_\bchi \fH$. The
scalar product on $\hfH$ is defined as
\begin{align} \label{eq-sharpsmooth-03}
\left\la \begin{pmatrix} f \\ g \end{pmatrix} \left| 
\begin{pmatrix} f' \\ g' \end{pmatrix} \right. \right\ra 
\ := \ 
\la f | f' \ra + \la g | g' \ra \, .
\end{align}
We further define the natural embedding $J: \fH \to \hfH$ by 
\begin{align} \label{eq-sharpsmooth-04}
J[\phi] \ := \ \begin{pmatrix} \chi \phi \\ \bchi \phi \end{pmatrix} 
\end{align}
and observe that the adjoint operator $J^*: \hfH \to \fH$ is given by
\begin{align} \label{eq-sharpsmooth-05}
J^* \left[ \begin{pmatrix} f \\ g \end{pmatrix} \right] 
\ = \ \chi f + \bchi g \, .
\end{align}
It easy to check that $J$ is an isometry and that 
$J^* J \ = \ \bfone_\fH$, which implies that 
$(J J^*) J \ = \ J (J^* J) \ = \ J$ and hence in summary
\begin{align} \label{eq-sharpsmooth-06}
J^* J \ = \ \bfone_\fH  
\quad \text{and} \quad
J J^* \big|_{\hfH_J} \ = \ \bfone_{\hfH_J} \, ,
\end{align}
where
\begin{align} \label{eq-sharpsmooth-07}
\hfH_J \ := \ \Ran(J) \ \subseteq \ \hfH \, .
\end{align}
%
\begin{lemma} \label{lem-01} 
$\hfH_J \subseteq \hfH$ is a closed subspace and $J: \fH \to \hfH_J$
is unitary.
\begin{proof}
It is clear that $\hfH_J \subseteq \hfH$ is a subspace. We prove
that, in general, the range $\Ran(\cI) \subseteq \hfH$ of any isometry
$\cI: \fH \to \hfH$ is closed. To this end, we assume $\hpsi \in \hfH$
to be an accumulation point of $\Ran(\cI)$ and 
$(\phi_n)_{n=1}^\infty \in \fH^\NN$ be a sequence such that 
$\cI \phi_n \to \hpsi$, as $n \to \infty$. Then, using that 
$\cI^* \cI = \bfone_\fH$ and that $\|\cI^*\|_\op = \|\cI\|_\op = 1$,
we obtain
\begin{align} \label{eq-sharpsmooth-08}
\| \phi_n - \cI^* \hpsi \|_{\fH}
\ = \ 
\| \cI^* ( \cI \phi_n - \hpsi ) \|_{\fH}
\ \leq \ 
\| \cI \phi_n - \hpsi \|_{\hfH} \, ,
\end{align}
and hence $(\phi_n)_{n=1}^\infty$ converges to $\cI^* \hpsi$. By
continuity of $\cI$, it follows that
\begin{align} \label{eq-sharpsmooth-09}
\hpsi 
\ = \ 
\lim_{n \to \infty} \cI \phi_n
\ = \ 
\cI [ \cI^* \hpsi ] 
\end{align}
belongs to the range of $\cI$, indeed. Hence, as a closed subspace, 
$\hfH_J \subseteq \hfH$ is itself a Hilbert space and $J$ is
unitary by \eqref{eq-sharpsmooth-06}.
\end{proof}
\end{lemma}
%
Next we recall from Hypothesis~\ref{hypo-T} that $T: \fD \to \fH$ is
assumed to be a closed operator, which commutes $\chi$ and $\bchi$ in
the sense that
\begin{align} \label{eq-sharpsmooth-10}
\chi \, T \ \subset \ T \, \chi \, , \quad 
\bchi \, T \ \subset \ T \, \bchi \, . 
\end{align}
Furthermore, $T_\chi: P_\chi \fD \to \fH_\chi$ and 
$T_\bchi: P_\bchi \fD \to \fH_\bchi$ denote the restrictions of $T$ to
$\fH_{\chi}$ and $\fH_{\bchi}$, and we recall from \eqref{eq-sharpsmooth-14}
that $T_{\bchi}$ is assumed to be bounded invertible on $\fH_{\bchi}$.

Now let $H$ be an operator on $\fH$ in the domain of $\cF_{\chi,T}$,
i.e., $H: \fD \to \fH$ is a closed linear operator, 
$H_\bchi \ = \ T_\bchi + W_\bchi$ is bounded invertible on $\fH_\bchi$,
and $\bchi H_\bchi^{-1} \bchi W \chi: \fH_\chi \to \fH_\bchi$
is bounded, where $W_\bchi = \bchi W \bchi$ and $W := H-T$ (we 
abbreviate $H_\bchi := H_{\chi,T}$ and $W := W_T$).

We now introduce the operators $\hT, \hW, \hH: \hfD \to \hfH$ by
\begin{align} \label{eq-sharpsmooth-11}
\hT \ := \ 
\begin{pmatrix} T_\chi & 0 \\ 0 & T_\bchi \end{pmatrix} 
\, , \quad
\hW \ := \ 
\begin{pmatrix} \chi W \chi & \chi W \bchi \\ 
                \bchi W \chi & \bchi W \bchi \end{pmatrix} 
\, , \quad
\hH \ := \ \hT + \hW \, , 
\end{align}
where $\hfD := P_\chi \fD \oplus P_\bchi \fD \subseteq \hfH$.
Since $T_{\chi}$ and $T_{\bchi}$ are closed, so is $\hT$.
We observe that if $\phi \in \dom(\hT J)$ then $J \phi \in \hfD$,
so $J \phi = f \oplus g$ with $f = P_\chi f \in \fD$ and 
$g = P_\bchi \, g \in \fD$. Since $\chi \fD, \bchi \fD \subseteq \fD$,
this implies that $J^*(f \oplus g) = \chi f + \bchi g \in \fD$.
Consequently,  
\begin{align} \label{eq-sharpsmooth-12}
\phi \ = \ J^* J \phi \ = \ J^*(f \oplus g) \ \in \ \fD \, .
\end{align}
from which we obtain that $\dom(\hT J) \subseteq \fD$ and similarly
also $\dom(\hH J) \subseteq \fD$. Conversely, if $\phi \in \fD$ then
\begin{align} \label{eq-sharpsmooth-13}
\hT J \phi 
\ = \ &
\begin{pmatrix} 
\chi T \phi \\ 
\bchi T \phi 
\end{pmatrix} 
\ = \ 
J T \phi \, , 
\\ \label{eq-sharpsmooth-14}
\hH J \phi 
\ = \ &
\begin{pmatrix} 
(T_\chi + \chi W \chi ) \chi \phi \, + \, (\chi W \bchi) \bchi \phi \\ 
(\bchi W \chi) \chi \phi  \, + \, (T_\bchi + \bchi W \bchi ) \bchi \phi
\end{pmatrix} 
\ = \ 
J H \phi \, ,
\end{align}
and hence $\phi \in \dom(\hT J) \cap \dom(\hH J)$. This and
\eqref{eq-sharpsmooth-12} imply that 
$\dom(\hT J) = \dom(\hH J) = \fD$. Moreover, by unitarity
of $J$, we obtain $J^* \hT J = T$ and  
\begin{align} \label{eq-sharpsmooth-15}
J^* \, \hH \, J \ = \ H
\quad \text{on $\fD$.}
\end{align}
It follows that $H$ is isospectral to $\hH$, because $J$ is unitary.
Now we apply the sharp Feshbach--Schur map with projections 
\begin{align} \label{eq-sharpsmooth-16}
\hP \ := \ \begin{pmatrix} \bfone_{\fH_\chi} & 0 \\ 0 & 0 \end{pmatrix} 
\quad \text{and} \quad
\hP^\perp \ := \ \begin{pmatrix} 0 & 0 \\ 0 & \bfone_{\fH_\bchi} \end{pmatrix}
\, , 
\end{align}
and obtain
\begin{align} \label{eq-sharpsmooth-17}
\cF_{\hP}(\hH) 
\ = \ &
\hP \, \hH \, \hP - \hP \, \hH \, \hP^\perp 
\big( \hP^\perp \hH \hP^\perp \big)^{-1} \hP^\perp \, \hH \, \hP
\nonumber \\[1ex]
\ = \ &
T_\chi \, + \, \chi \, W \, \chi \, - \,
\chi \, W \, \bchi ( H_\bchi )^{-1} \bchi \, W \chi 
\ = \ 
\cF_{\chi, T}(H) 
\end{align}
on $\fD_\chi$.

\section{Proofs of Sect.~\ref{RF}: Renormalization Flow}
\label{sec:RG}

\subsection{The Rescaled Smooth Feshbach--Schur Map}
%
In the following, we apply the results from Section~\ref{sec:SFM} with
the choices $T := H_\ph$ and $S := w_{0,0}(H_\ph)$. We recall that the
resolvent set $\rho(A)$ of an operator $A$ consists of all complex
numbers $\lambda$, for which $A - \lambda \bfone$ is bounded
invertible. We further recall from Definition~\ref{Delta}
the definition of $\Delta_{\chi,T}(S)$ and, specifically, 
$\Delta_{\chi_\alpha, H_\ph}[w_{0,0}(H_\ph)] 
= H_\ph \chi_\alpha^2 + w_{0,0}(H_\ph) \bchi_\alpha^2$.
%
\begin{theorem}\label{thm-semigroup}
Let $\uw \in \cW^\xi$ and $\alpha, \beta \geq 0$. Suppose that 
$\qHH[\uw] \in \dom(\mathrm{F}_{\chi_\alpha, H_\ph})$ and $0\in \rho(\qHH[\uw])$. 
Further assume that 
$\hcR_\alpha(\qHH[\uw]) \in \dom(\mathrm{F}_{\chi_\beta, H_\ph})$
and that the restriction 
$\Delta_{\chi_\alpha, H_\ph}[w_{0,0}(H_\ph)]\big|_{ \fH_{\bchi_\alpha}} 
\in \cB(\fH_{\bchi_\alpha})$ of 
$\Delta_{\chi_\alpha, H_\ph}(w_{0,0}(H_\ph))$ to $\fH_{\bchi_\alpha}$
is bounded invertible, see Definition~\ref{Delta}. 
It follows that $\qHH[\uw]\in \dom(\mathrm{F}_{\chi_{\alpha+\beta}, H_\ph})$ 
and that $\hcR$ satisfies the semigroup property
\begin{align} \label{eq-Semigroup}
\hcR_\beta \big( \hcR_\alpha(\qHH[\uw]) \big) 
\ = \ 
\hcR_{\alpha+\beta}(\qHH[\uw]).
\end{align}  
\end{theorem}
\begin{proof} We prove the result in two steps, Step~1 and Step~2 below. 
In Step~1, we prove that $\qHH[\uw] \in \dom(\mathrm{F}_{\chi_\alpha, H_\ph})$ 
and $\hcR_\alpha(\qHH[\uw]) \in \dom(\mathrm{F}_{\chi_\beta, H_\ph})$ together 
imply that $\qHH[\uw]\in \dom(\mathrm{F}_{\chi_{\alpha+\beta}, H_\ph})$. 
Then in Step~2 we prove \eqref{eq-Semigroup}.

\paragraph*{Step~1: 
$\qHH[\uw]\in \dom(\mathrm{F}_{\chi_{\alpha+\beta}, H_\ph})$.} 
Given $\alpha, \beta \geq 0$, we 
define the functions, 
$X, \bX: [0, \infty) \to \RR$ by 
$X(r) := 1$  and $\bX(r) := 0$, whenever 
$\bchi_{\alpha + \beta}(r) = 0$, and 
\begin{align}\label{XoX}
X(r) \ := \ 
\frac{\bchi_\beta(e^{\alpha} r) \: \chi_\alpha(r)}{ \bchi_{\alpha + \beta}(r) }
\quad \text{and} \quad 
\bX(r) \ := \ \frac{ \bchi_\alpha(r) }{ \bchi_{\alpha+\beta}(r) } \, ,
\end{align}  
provided $\bchi_{\alpha + \beta}(r) > 0$. Note that $X$ and $\bX$ depend
on $\alpha$ and $\beta$, but we refrain from displaying this dependence
in the notation.
Recalling \eqref{eq-functionchisemigroup}, we deduce that 
$\chi_{\alpha+\beta}(r) = 1$ implies that $\chi_\alpha(r) = 1 $ and, 
therefore, $\bchi_{\alpha+\beta}(r) = 0$ implies that $\bchi_\alpha(r) = 0$.  
Hence,
\begin{align}\label{jor1}
\bchi_\alpha(r) = 0  \quad \Leftrightarrow \quad \bX(r) = 0 \, ,
\end{align}  
and thus 
\begin{align} \label{jor2}
\fH_{\bchi_{\alpha+\beta}} 
\ = \ &
\Ran\big[ \bchi_{\alpha+\beta}(H_\ph) \big] 
\\[1ex] \nonumber
\ \supseteq \ &
\fH_{\bchi_{\alpha}} 
\ = \ 
\Ran\big[ \bchi_\alpha(H_\ph) \big] 
\ = \ 
\Ran\big[ \bX(H_\ph) \big] 
\ = \ 
\fH_{\bX(H_\ph)} \, .
\end{align}  
Moreover, using again \eqref{eq-functionchisemigroup}, we obtain that
$\bchi_{\alpha+\beta}(r) = 0$ implies that $\chi_\alpha(r) = 1$ and
$\bchi_{\beta}(e^{\alpha} r) = 0$. From the above discussion we obtain
\begin{align}\label{Falta}
\bchi_{\beta}(e^{\alpha} r) \, \chi_\alpha(r) 
\ = \
X(r) \, \bchi_{\alpha + \beta}(r)
\quad \text{and} \quad 
\bchi_\alpha(r) \ = \ \bX(r) \, \bchi_{\alpha + \beta}(r) \, ,
\end{align}   
for all $r \geq 0$. Recall that 
$\chi_{\alpha+\beta}(r) = \chi_\beta(e^\alpha r) \chi_\alpha(r) $ and that 
$\chi_\rho^2(r) + \bchi_\rho^2(r) = 1$, for every $\rho > 0$ 
[see Eq.~\eqref{eq-functionchisemigroup}]. This implies that  
\begin{align}
\bchi_\beta^2(e^\alpha r) \, \chi_\alpha^2(r) \, + \, \bchi_\alpha^2(r)
\ = \ & 
\big( 1 - \chi_\beta^2(e^\alpha r) \big) \, \chi_\alpha^2(r) 
\, + \, 1 - \chi_\alpha^2(r)
\nonumber \\[1ex]
\ = \ & 
1 - \chi_{\alpha+\beta}^2(r) 
\ = \ 
\bchi_{\alpha+\beta}^2(r)  
\end{align}  
and, therefore, 
\begin{align}\label{X}
X^2(r) \, + \, \bX^2(r) \ = \ 1 \, ,
\end{align} 
for all $r \geq 0$. We identify, as usual in this paper,
\begin{align}
X \ \equiv \ X(H_\ph) 
\quad \text{and} \quad 
\bX \ \equiv \ \bX(H_\ph) \, . 
\end{align} 
Now, we introduce 
$\uw^{\alpha+\beta}  = ( w^{\alpha+\beta}_{m,n}  )_{m, n \geq 0 }  
\in \cW^\xi$ by $\uw_\gI^{\alpha+\beta} := 
\bchi_{\alpha+\beta} \, \uw_\gI \, \bchi_{\alpha+\beta}$
and $w^{\alpha+\beta}_{0,0}(r) := 
r \chi_{\alpha+\beta}^2(r) P_{\bchi_{\alpha+\beta}}
+ w_{0,0}(r) \bchi_{\alpha+\beta}^2(r)$, such that
\begin{align} \label{eq-wchidef}
\qHH[\uw]_{\bchi_{\alpha+\beta}, H_\ph} 
\ = \ 
H_\ph \, \chi_{\alpha+\beta}^2 \, P_{\bchi_{\alpha+\beta}} \, + \, 
\bchi_{\alpha+\beta} \, \qHH[\uw] \, \bchi_{\alpha+\beta} 
\ = \ 
\qHH[\uw^{\alpha+\beta}] \, , 
\end{align}  
recalling \eqref{Jpto1} and the definition of $\qHH[\uw]$ in 
\eqref{HW-1}. Additionally using \eqref{jor2}, it follows that
\begin{align} \label{EqCom}
\qHH[\uw^{\alpha+\beta}]_{\bX, H_\ph} 
\ = \ &  
H_\ph \, X^2 \, P_\bX \, + \, \bX \, \qHH[\uw^{\alpha+\beta}] \, \bX
\nonumber \\[1ex] 
\ = \ &  
H_\ph \big( X^2 \, P_\bX \, + \, \chi_{\alpha+ \beta}^2 \, \bX^2 \big) 
\, + \, \bX \, \bchi_{\alpha + \beta} \, \qHH[\uw] \, \bchi_{\alpha + \beta} \, \bX
\nonumber \\[1ex] 
\ = \ &  
H_\ph \big( X^2 \, P_\bX \, + \, (1- \bchi_{\alpha+ \beta}^2) \bX^2 \big) 
\, + \, 
\bX \, \bchi_{\alpha + \beta} \, \qHH[\uw] \, \bchi_{\alpha + \beta} \, \bX
\nonumber \\[1ex] 
\ = \ &  
H_\ph \big( P_\bX \, - \, \bX^2 \, \bchi_{\alpha+\beta}^2 \big) 
\, + \, 
\bX \, \bchi_{\alpha + \beta} \, \qHH[\uw] \, \bchi_{\alpha + \beta} \, \bX \, ,
\end{align}  
where we further use \eqref{X} and that 
$\chi^2_\gamma + \bchi^2_\gamma = 1$, 
for $\gamma \geq 0$ [see Eq.~\eqref{eq-functionchisemigroup}].
Eqs.~\eqref{Falta} and \eqref{EqCom} imply [recall \eqref{Jpto1}]
\begin{align}\label{nomams}
\qHH[\uw]_{\bchi_\alpha, H_\ph} 
\ = \ 
H_\ph \, \chi_\alpha^2 \, P_\bX 
\, + \, \bchi_\alpha \, \qHH[\uw] \, \bchi_\alpha 
\ = \ 
\qHH[\uw^{\alpha+\beta}]_{\bX, H_\ph} \, , 
\end{align}  
where we use that $P_{\bchi_\alpha} = P_\bX$, due to \eqref{jor2}.
Similarly, substituting $\uw^{\alpha+\beta} - \uw_\gI^{\alpha+\beta}$
for $\uw^{\alpha+\beta}$ in \eqref{EqCom}, we obtain (recall
Definition~\ref{Delta})
\begin{align} \label{eq-freeXrelation-1}
\Delta_{\chi_\alpha, H_\ph}(w_{0,0}(H_\ph)) 
\ = \ & 
H_\ph \, \chi_\alpha^2 \, P_{\bchi_\alpha} 
\, + \, w_{0,0}(H_\ph) \, \bchi_\alpha^2
\nonumber \\[1ex] 
\ = \ &  
H_\ph \, X^2 \, P_\bX \, + \, 
\big[ H_\ph \, \chi_{\alpha+\beta}^2 \, + \, 
w_{0,0}(H_\ph) \, \bchi_{\alpha+\beta}^2 \big] \, \bX^2 
\nonumber \\[1ex] 
\ = \ &  
\Delta_{X, H_\ph}(w^{\alpha+\beta}_{0,0}(H_\ph)) \, , 
\end{align}  
using that 
$w^{\alpha+\beta}_{0,0}(H_\ph) = 
H_\ph \; \chi_{\alpha+\beta}^2 + w_{0,0}(H_\ph) \, \bchi_{\alpha+\beta}^2$.
We recall \eqref{f} and write
\begin{align} \label{imp}
f_{\chi_\alpha, H_\ph}[w_{0,0}(H_\ph)] 
\ = \ & 
\big( H_\ph \, \Delta_{\chi_\alpha, H_\ph}[ w_{0,0}(H_\ph) ]^{-1} \big) 
\big|_{\fH_{\bchi_\alpha}}  \oplus \bfone_{\fH_{\bchi_\alpha}^\perp}
\nonumber \\[1ex]
\ = \ & 
\big( H_\ph \, \Delta_{X, H_\ph}[ w_{0,0}^{\alpha+\beta}(H_\ph) ]^{-1} \big) 
\big|_{\fH_\bX}  \oplus \bfone_{\fH_\bX^\perp} 
\nonumber \\[1ex]
\ = \ & 
f_{X, H_\ph}[w^{\alpha+\beta}_{0,0}(H_\ph)] \, ,   
\end{align}  
thanks to $\fH_{\bchi_\alpha} = \fH_\bX$, see \eqref{jor2}.

Now, by \eqref{Jpto1} and \eqref{eq-functionchisemigroup}, taking
$\hcR_\alpha(\qHH[\uw])$ instead of $\qHH[\uw]$, we obtain 
\begin{align} \label{renn}
\hcR_\alpha(\qHH[\uw])_{\bchi_\beta, H_\ph} 
\ = \ & 
H_\ph \, P_{\bchi_\beta} \, + \, 
\bchi_\beta \big( \hcR_\alpha(\qHH[\uw]) - H_\ph \big) \bchi_\beta  
\nonumber \\[1ex]
\ = \ & 
H_\ph \, \chi_{\beta}^2 \, P_{\bchi_\beta} \, + \, 
\bchi_\beta \big( \hcR_\alpha(\qHH[\uw]) \big) \bchi_\beta\ \, .
\end{align} 
We recall our assumption that 
$\hcR_\alpha(\qHH[\uw]) \in \dom(\mathrm{F}_{\chi_\beta, H_\ph})$, which
implies that $\hcR_\alpha(\qHH[\uw])_{\bchi_{\beta}, H_\ph}$ is bounded
invertible [see the text around \eqref{Jpto1}]. The first step of the
proof is accomplished if we prove that 
$\qHH[\uw^{\alpha+\beta}] = \qHH[\uw]_{\bchi_{\alpha + \beta}, H_\ph}$
is bounded invertible [see \eqref{eq-wchidef})], because all operators
we consider are bounded. In view of Proposition~\ref{prop-isospecSFB},
this follows from the invertibility of $F_{X, H_\ph}(
\qHH[\uw^{\alpha+\beta}] )$ which, in turn, is a consequence of
the identity
\begin{align} \label{dom}
\hcR_\alpha \big( \qHH[\uw] \big)_{\bchi_{\beta}, H_\ph}
\ = \ 
e^\alpha \, \Gamma_\alpha \, 
F_{X, H_\ph} \big(\qHH[\uw^{\alpha+\beta}] \big) \, \Gamma_\alpha^* 
\end{align}  
proven below, together with the fact that
$\qHH[\uw^{\alpha+\beta}]$ belongs to the domain of $F_{X, H_\ph}$.

Using that, for every measurable function $g : \RR \to \CC$, we have
that $\Gamma_\alpha g(H_\ph) \Gamma_\alpha^* = g(e^{-\alpha}H_\ph)$,
we obtain
\begin{align} \label{porque}
e^\alpha \, \Gamma_\alpha & \, H_\ph \, \Gamma_\alpha^* \, \chi_\beta^2
\, + \, \bchi_\beta \, e^\alpha \, \Gamma_\alpha \, 
F_{\chi_\alpha, H_\ph}\big( \qHH[\uw] \big) \, \Gamma_\alpha^* \, \bchi_\beta
\\[1ex] \nonumber
\ = \ &
e^\alpha \, \Gamma_\alpha \Big[ 
H_\ph \, \Gamma_\alpha^* \, \chi_\beta^2 \, \Gamma_\alpha 
\, + \, \Gamma_\alpha^* \, \bchi_\beta \, \Gamma_\alpha \,
F_{\chi_\alpha, H_\ph}\big( \qHH[\uw] \big) \, \Gamma_\alpha^* \, \bchi_\beta \, 
\Gamma_\alpha \Big] \Gamma_\alpha^*
\\[1ex] \nonumber
\ = \ &
e^\alpha \, \Gamma_\alpha \Big[ H_\ph \, \chi_\beta(e^\alpha H_\ph)^2 
\, + \, \bchi_\beta(e^\alpha H_\ph) \, 
F_{\chi_\alpha, H_\ph}\big( \qHH[\uw] \big) \, \bchi_\beta(e^\alpha H_\ph) \Big]
\Gamma_\alpha^* \, .
\end{align}  
Using Eq.~\eqref{maintsyeah} and 
$\qHH[\uw] = w_{0,0}(H_\ph) + \qWW[\uw]$, 
we analyze the second term in the last line which contains the smooth 
Feshbach--Schur map,
\begin{align}
\bchi_\beta & (e^\alpha H_\ph) \,
F_{\chi_\alpha, H_\ph}\big( \qHH[\uw] \big) \, \bchi_\beta(e^\alpha H_\ph) 
\nonumber \\[1ex] 
\ = \ &
w_{0,0}(H_\ph) \, f_{\chi_\alpha, H_\ph}[w_{0,0}(H_\ph)] \, 
\bchi_\beta(e^\alpha H_\ph)^2 
\\[0,5ex] \nonumber &
+ \, 
\bchi_\beta(e^\alpha H_\ph) \, \chi_\alpha \, 
f_{\chi_\alpha, H_\ph}[w_{0,0}(H_\ph)] \, \qWW[\uw] \, 
f_{\chi_\alpha, H_\ph}[w_{0,0}(H_\ph)] \, \chi_\alpha \, \bchi_\beta(e^\alpha H_\ph)
\\[0,5ex] \nonumber &
- \, 
\bchi_\beta(e^\alpha H_\ph) \, \chi_\alpha \, 
f_{\chi_\alpha, H_\ph}[w_{0,0}(H_\ph)] \, \qWW[\uw] \, \bchi_\alpha
\\ \nonumber &
\qquad H_{\bchi_\alpha, H_\ph}(\uw)^{-1} \, \bchi_\alpha \, \qWW[\uw] \,
f_{\chi_\alpha, H_\ph}[w_{0,0}(H_\ph)] \, 
\chi_\alpha \, \bchi_\beta(e^\alpha H_\ph) \, .
\end{align} 
Additionally using Eq.~\eqref{Falta}, we obtain:  
\begin{align}  \label{casi}
\bchi_\beta(e^\alpha H_\ph) \, & 
F_{\chi_\alpha, H_\ph}\big( \qHH[\uw] \big) \, \bchi_\beta(e^\alpha H_\ph) 
\nonumber \\[1ex] 
\ = \ &
w_{0,0}(H_\ph) \, f_{\chi_\alpha, H_\ph}[w_{0,0}(H_\ph)] \, 
\bchi_\beta(e^\alpha H_\ph)^2 
\\[0,5ex] \nonumber &
+ \, 
X \, f_{\chi_\alpha, H_\ph}[w_{0,0}(H_\ph)] \, 
\bchi_{\alpha+\beta} \, \qWW[\uw] \, \bchi_{\alpha+\beta} \, 
f_{\chi_\alpha, H_\ph}[w_{0,0}(H_\ph)] \, X 
\\[0,5ex] \nonumber &
- \, 
X \, f_{\chi_\alpha, H_\ph}[w_{0,0}(H_\ph)] \, 
\bchi_{\alpha+\beta} \, \qWW[\uw] \, \bchi_{\alpha+\beta} \, 
\\ \nonumber &
\qquad 
\bX \, H_{\bchi_\alpha, H_\ph}(\uw)^{-1} \, \bX \, 
\bchi_{\alpha+\beta} \, \qWW[\uw] \, \bchi_{\alpha+\beta} \, 
f_{\chi_\alpha, H_\ph}[w_{0,0}(H_\ph)] \, X \, .
\end{align} 
Moreover, recalling \eqref{HW-1} and \eqref{HW-2}, we observe that 
\begin{align}  \label{casi-2}
\bchi_{\alpha+\beta} \, \qWW[\uw] \, \bchi_{\alpha+\beta} 
\ = \ 
\qWW[\uw^{\alpha+\beta}] \, ,
\end{align} 
which, together with \eqref{eq-wchidef}, \eqref{nomams}, \eqref{imp}, 
\eqref{casi}, and 
\begin{align}  \label{casi-3}
w_{0,0}^{\alpha+\beta}(H_\ph) 
\ = \ 
H_\ph \, \chi_{\alpha+\beta}^2(H_\ph) \, + \, 
w_{0,0}(H_\ph) \, \bchi_{\alpha+\beta}^2(H_\ph) 
\end{align} 
implies that
\begin{align}  \label{casip}
\bchi_\beta(e^\alpha H_\ph) \, & 
F_{\chi_\alpha, H_\ph}\big( \qHH[\uw] \big) \, \bchi_\beta(e^\alpha H_\ph) 
\nonumber \\[1ex] 
\ = \ &
w_{0,0}(H_\ph) \, f_{X, H_\ph}[w_{0,0}^{\alpha+\beta}(H_\ph)] \, 
\bchi_\beta(e^\alpha H_\ph)^2 
\nonumber \\[0,5ex] &
+ \, 
X \, f_{X, H_\ph}[w_{0,0}^{\alpha+\beta}(H_\ph)] \, \qWW[\uw^{\alpha+\beta}] \, 
f_{X, H_\ph}[w_{0,0}^{\alpha+\beta}(H_\ph)] \, X \, 
\nonumber \\[0,5ex] &
- \, 
X \, f_{X, H_\ph}[w_{0,0}^{\alpha+\beta}(H_\ph)] \, \qWW[\uw^{\alpha+\beta}] \, 
\\ \nonumber &
\qquad 
\bX \, H_{\bchi_\alpha, H_\ph}(\uw)^{-1} \, \bX \, 
\qWW[\uw^{\alpha+\beta}] \, f_{X, H_\ph}[w_{0,0}^{\alpha+\beta}(H_\ph)] \, X 
\\[1ex] \nonumber 
\ = \ &
w_{0,0}(H_\ph) \, f_{X, H_\ph}[ w_{0,0}^{\alpha+\beta}(H_\ph) ] \, 
\bchi_\beta(e^\alpha H_\ph)^2  
\\ \nonumber &
\qquad 
- \, w_{0, 0}^{\alpha+\beta}(H_\ph) \, f_{X, H_\ph}[w^{\alpha+\beta}_{0,0}(H_\ph)]
\, + \, F_{X, H_\ph}\big( \qHH[\uw^{\alpha+\beta}] \big) \, ,   
\end{align} 
where, in the last step, we apply \eqref{maintsyeah} again, with 
$\chi = X$, $S = \uw_{0,0}^{\alpha+\beta}(H_\ph)$, and $T = H_\ph$.

Using \eqref{f}, we observe that [see \eqref{imp} and the text 
below \eqref{XoX}] 
\begin{align} \label{homo}
f_{X, H_\ph} & [ w_{0,0}^{\alpha+\beta}(H_\ph) ] \, 
\Delta_{\chi_\alpha, H_\ph}[ w_{0,0}(H_\ph) ] 
\\[1ex] \nonumber
\ = \ &
f_{\chi_\alpha, H_\ph}[ w_{0,0}(H_\ph) ] \, 
\Delta_{\chi_\alpha, H_\ph}[ w_{0,0}(H_\ph) ] 
\ = \ H_\ph \, ,   
\end{align}  
where we use that 
$\Delta_{\chi_\alpha, H_\ph}[w_{0,0}(H_\ph)]\big|_{\fH_{\bchi_\alpha}^\perp} 
= H_\ph \big|_{\fH_{\bchi_\alpha}^\perp}$. 
Moreover, \eqref{eq-wchidef} implies that 
$w_{0,0}^{\alpha+\beta}(H_\ph) = 
\Delta_{\chi_{\alpha + \beta}, H_\ph}[ w_{0,0}(H_\ph) ]$, see Definition~\ref{Delta}.  
Consequently, we obtain
\begin{align} \label{homo1}
& w_{0,0}(H_\ph) \, f_{X, H_\ph}[w_{0,0}^{\alpha+\beta}(H_\ph)] \, 
\bchi_\beta(e^\alpha H_\ph)^2 
\nonumber \\ & \qquad 
\, - \, 
w_{0,0}^{\alpha+\beta}(H_\ph) \, f_{X, H_\ph}[ w^{\alpha+\beta}_{0,0}(H_\ph) ]
\\[1ex] \nonumber
& \ = \ 
\big\{ w_{0,0}(H_\ph) \, \bchi_\beta(e^\alpha H_\ph)^2 \, - \,    
\Delta_{\chi_{\alpha + \beta}, H_\ph}[ w_{0,0}(H_\ph) ] \big\} \,
f_{X, H_\ph}[w^{\alpha+\beta}_{0,0}(H_\ph)] \, .
\end{align}  
Moreover, using Definition~\ref{Delta} and 
\eqref{eq-functionchisemigroup}, we get
\begin{align} \label{nnoo}
w_{0,0} & (H_\ph) \, \bchi_\beta(e^\alpha H_\ph)^2    
\, - \, \Delta_{\chi_{\alpha + \beta}, H_\ph}(w_{0,0}(H_\ph))  
\\[1ex] \nonumber
\ = \ & 
w_{0,0}(H_\ph) \big\{ [ 1 - \chi_\beta(e^\alpha H_\ph)^2 ] 
- [ 1- \chi_\beta(e^\alpha H_\ph)^2 \chi_{\alpha }^2 ] \big\}  
\, - \, H_\ph \, \chi_{\alpha + \beta}^2  
\\[1ex] \nonumber
\ = \ & 
- w_{0,0}(H_\ph) \, \chi_\beta(e^\alpha H_\ph)^2 \, [1 - \chi_\alpha^2] 
\, - \, H_\ph \, \chi_{\alpha + \beta}^2  
\\[1ex] \nonumber
\ = \ & 
- \chi_\beta(e^\alpha H_\ph)^2 \big( H_\ph \, \chi_{\alpha }^2 \, + \, 
w_{0,0}(H_\ph) \, \bchi_\alpha^2 \big) 
\\[1ex] \nonumber
\ = \ & 
- \chi_\beta(e^\alpha H_\ph)^2 \, \Delta_{\chi_\alpha, H_\ph}[ w_{0,0}(H_\ph) ]  \, . 
\end{align}  
Eqs.~\eqref{casip}, \eqref{homo}, \eqref{homo1} and \eqref{nnoo}  
lead us to
\begin{align}\label{casiend}
\bchi_\beta(e^\alpha H_\ph) \, & F_{\chi_\alpha, H_\ph}\big( \qHH[\uw] \big) \, 
\bchi_\beta(e^\alpha H_\ph) 
\\[1ex] \nonumber
\ = \ &
- \chi_\beta(e^\alpha H_\ph)^2 \, H_\ph \, + \, 
F_{X, H_\ph}\big( \qHH[\uw^{\alpha+\beta}] \big) \, . 
\end{align}  
From Definition~\ref{def:rescaledfsmap}, Eqs.~\eqref{porque} and 
\eqref{casiend}, and observing that 
$e^\alpha \Gamma_\alpha H_\ph\Gamma_\alpha^* = H_\ph$, 
we obtain 
\begin{align} \label{coyo}
H_\ph \, \chi_\beta^2 \, + \, 
\bchi_\beta \, \hcR_\alpha\big( \qHH[\uw] \big) \, \bchi_\beta 
\ = \
e^\alpha \, \Gamma_\alpha \, 
F_{X, H_\ph}( \qHH[\uw^{\alpha+\beta}] )\Gamma_\alpha^* \, .  
\end{align}  
Eqs.~\eqref{renn} and \eqref{coyo} imply the desired identity, namely,
Eq.~\eqref{dom}.

\paragraph*{Step 2: Flow Property, Eq.~\eqref{eq-Semigroup}.}
We turn to the proof of \eqref{eq-Semigroup}. Since $0\in
\rho(\qHH[\uw])$, Proposition~\ref{prop-isospecSFB} shows that
$F_{\chi_\alpha, H_\ph}(\qHH[\uw])$, $\hcR_\alpha(\qHH[\uw])$, and
$\hcR_\beta(\hcR_\alpha(\qHH[\uw]))$ are bounded invertible. Moreover,
$\hcR_{\alpha+\beta}(\qHH[\uw])$ is bounded invertible as well, by
Step~1. We recall that $(\Gamma_\alpha)_{\alpha \in \RR}$ is a group
of unitary operators and that for every measurable function $g : \RR
\to \CC$, $\Gamma_\alpha g(H_\ph) \Gamma_\alpha^* = g(e^{-\alpha }
H_\ph)$ and, in particular, that $e^\alpha \Gamma_\alpha H_\ph
\Gamma_\alpha^* = H_\ph$. Now, from \eqref{eq-Fsmoothinvers} and
Definition~\ref{def:rescaledfsmap}, it follows that 
\begin{align} \label{eq-step2-01}
\big[ \hcR_\beta \big( \hcR_\alpha & (\qHH[\uw]) \big) \big]^{-1} 
\ = \  
e^{-\beta} \, \Gamma_\beta
\big[\chi_\beta \, \hcR_\alpha(\qHH[\uw])^{-1} \, \chi_\beta 
\, + \, \bchi_\beta \, H_\ph^{-1}\bchi_\beta \, \big] \Gamma_\beta^*
\nonumber \\[1ex]
\ = \ & 
e^{-\beta} \, \Gamma_\beta \, \big[ \chi_\beta \, e^{-\alpha} \, \Gamma_\alpha 
\, \big\{ \chi_\alpha \, \qHH[\uw]^{-1} \, \chi_\alpha \, + \, 
\bchi_\alpha \, H_\ph^{-1} \, \bchi_\alpha \big\} \Gamma_\alpha^* \, \chi_\beta
\nonumber \\
& \qquad \qquad 
+ \bchi_\beta \, e^{-\alpha} \, \Gamma_\alpha \, H_\ph^{-1} \, 
\Gamma_\alpha^* \, \bchi_\beta \big] \Gamma_\beta^*
\nonumber \\[1ex]
\ = \ & 
e^{-\alpha-\beta} \, \Gamma_\beta \, \Gamma_\alpha \, 
\big[ \chi_\beta(e^\alpha H_\ph) \, 
\big\{ \chi_\alpha \, \qHH[\uw]^{-1} \, \chi_\alpha \, + \, 
\bchi_\alpha \, H_\ph^{-1} \, \bchi_\alpha \big\} \chi_\beta(e^\alpha H_\ph)
\nonumber \\
& \qquad \qquad 
+ \bchi_\beta(e^\alpha H_\ph) \, H_\ph^{-1} \, \bchi_\beta(e^\alpha H_\ph) \big] 
\Gamma_\alpha^* \, \Gamma_\beta^*
\nonumber \\[1ex]
\ = \ & 
e^{-\alpha-\beta} \, \Gamma_{\alpha+\beta} \, 
\big[ \chi_\beta(e^\alpha H_\ph) \, \chi_\alpha \, \qHH[\uw]^{-1} \, 
\chi_\alpha \, \chi_\beta(e^\alpha H_\ph)
\nonumber \\
& \qquad \qquad 
+ H_\ph^{-1} \big\{ (1- \chi_\alpha^2) \, \chi_\beta(e^\alpha H_\ph)^2 
\, + \, \bchi_\beta(e^\alpha H_\ph)^2 \big\} \big] \Gamma_{\alpha+\beta}^*
\nonumber \\[1ex]
\ = \ & 
e^{-\alpha-\beta} \, \Gamma_{\alpha+\beta} \, 
\big[ \chi_{\alpha+\beta} \, \qHH[\uw]^{-1} \, \chi_{\alpha+\beta} 
\, + \, H_\ph^{-1} \, \bchi_{\alpha+\beta}^2 \big] \Gamma_{\alpha+\beta}^*
\nonumber \\[1ex]
\ = \ & 
\hcR_{\alpha+\beta}(\qHH[\uw])^{-1} \, ,
\end{align}  
where we use \eqref{eq-functionchisemigroup}.  
\end{proof}

\subsection{Renormalization of the Spectral Parameter}
%
In the following lemma we establish the invertibility of the
restrictions of $\Delta_{\chi_{\alpha }, H_\ph}[w_{0,0}(H_\ph)]$ and
$\qHH[\uw]_{\bchi_\alpha, H_\ph}$ to $\fH_{\bchi_\alpha}$. For this we
recall Definition~\ref{Delta}, and Eqs.~\eqref{Jpto1}, \eqref{vI-01},
\eqref{yeahhh}, and \eqref{ur} and the notation $\cB(\fE)$ for the
Banach space of bounded linear operators on a Hilbert space $\fE$.
%
\begin{lemma}\label{lemma: domFchiconditions}
Suppose that $|z | \leq \frac{1}{4} e^{-\alpha}$, that 
$\|\uw - \ur\|_{(\partial_r)} < \frac{1}{2}$, and that
$4 \xi \| \uw_\gI \|^{(\xi)} e^{\alpha} < 
1 - 2 \|\uw - \ur\|_{(\partial_r)}$.
Then, the restrictions of 
$\Delta_{\chi_{\alpha }, H_\ph}[w_{0,0}(H_\ph)]$
and $\qHH[\uw]_{\bchi_\alpha, H_\ph}$ to $\fH_{\bchi_\alpha}$ are invertible, and
these inverses obey the norm bounds
\begin{align}\label{H-0}
\Big\| \big( 
\Delta_{\chi_\alpha, H_\ph}[w_{0,0}(H_\ph)] \big|_{\fH_{\bchi_\alpha}} 
\big)^{-1} \Big\|_{\mathcal{B}(\fH_{\bchi_\alpha})} 
\ \leq \ 
\frac{4 \, e^{\alpha}}{1 - 2\|\uw - \ur\|_{(\partial_r)}} 
\end{align}  
and
\begin{align}\label{H-1}
\big\| ( \qHH[\uw]_{\bchi_\alpha, H_\ph} )^{-1} \big\|_{\mathcal{B}(\fH_{\bchi_\alpha})} 
\ \leq \ 
\frac{4 e^{\alpha}}{1 - 2\|\uw - \ur\|_{(\partial_r)}
- 4 \xi \, \| \uw_\gI \|^{(\xi)} \, e^{\alpha} } \, . 
\end{align}  
Moreover, $\qHH[\uw] \equiv \qHH[\uw(z)]$ is invertible 
provided $\rRe(z) \geq \xi \| \uw_\gI \|^{ (\xi)} 2 e^{\alpha}$. 
\begin{proof}
In this proof we denote by $\delta_z : [0,1] \to \CC$ the function
\begin{align}
\delta_z(r) \ := \ 
r \, \chi^2_\alpha(r) \, + \, w_{0,0}(z,r) \, \bchi^2_\alpha(r)  
\end{align}  
such that  
\begin{align}\label{EM0}
\Delta_{\chi_{\alpha}, H_\ph}[w_{0,0}(H_\ph)] \ =: \ \delta_{z}(H_\ph),  
\end{align}  
according to Definition~\ref{Delta}. 

We recall that $w_{0,0}(z,0) = z$, see \eqref{w00z}. Assuming
$e^{-\alpha}/2 \leq r \leq 1$, we observe that $|z| \leq r/2$ and
hence:
\begin{align} \label{EM1}
|\delta_{z} (r)| 
\ = \ & 
\big| r \, + \, \bchi_\alpha^2(r) \, 
[w_{0,0}(z,r) - r - w_{0,0}(z,0) + z] \big| 
\nonumber \\[1ex]
\ \geq \ & 
r \, - \, r \, \Big|\frac{w_{0,0}(z,r) - w_{0,0}(z,0)}{r} - 1 \Big| \, - \, |z|
\nonumber \\[1ex]
\ \geq \ & 
r \, \Big( 1 - 
\sup_{r \in [e^{-\alpha}/2, 1]} |\partial_r w_{0,0}(z,r) - 1| \Big) 
\, - \, \frac{r}{2}
\\[1ex] \nonumber 
\ \geq \ & 
r \, \Big( \frac{1}{2} - \|\uw - \ur\|_{(\partial_r)} \Big) 
\ > \ 
\frac{r}{2} \, \big( 1 - 2 \|\uw - \ur\|_{(\partial_r)} \big) 
\ > \ 0 \, ,
\end{align}  
where we use \eqref{eq:muNormtilde-01,6} and the hypotheses of the
present lemma.

Next, we establish \eqref{H-0}. Eqs.~\eqref{EM0} and \eqref{EM1}
together with the functional calculus imply that
\begin{align}\label{carbon}
\big\| \Delta_{\chi_{\alpha}, H_\ph}[w_{0,0}(H_\ph)]^{-1} 
\big\|_{\mathcal{B}(\fH_{\bchi_\alpha})}  
\ = \ & 
\sup_{r \in [e^{-\alpha}/2, 1]} \Big| \frac{1}{\delta_{z}(r)} \Big| 
\ \leq \ 
\frac{4 e^{\alpha}}{1 - 2\|\uw - \ur\|_{(\partial_r)}} \, ,
\end{align}  
since $e^{-\alpha}/2 \leq H_\ph \leq 1$ on
$\fH_{\bchi_\alpha}$ according to \eqref{eq-functionchisemigroup}. 
This proves \eqref{H-0}.  

We turn to \eqref{H-1}. It follows from \eqref{opbound} that 
\begin{align}\label{fep1}
\| \qWW[\uw] \| \ \leq \ \xi \, \| \uw_\gI \|^{(\xi)}
\end{align}  
and, therefore,
\begin{align} \label{NYC1} 
\Big\| \bchi_\alpha \, \qWW[\uw] \, \bchi_\alpha \, 
\big( \Delta_{\chi_\alpha, H_\ph}[w_{0,0}(H_\ph)] |_{\fH_{\bchi_\alpha}} \big)^{-1} 
\Big\|_{\mathcal{B}(\fH_{\bchi_\alpha})} 
\ \leq \ &  
\frac{ 4 e^{\alpha} \, \xi \, \| \uw_\gI \|^{(\xi)} }{
1 - 2\|\uw - \ur\|_{(\partial_r)} } 
\ < \ 1 \, . 
\end{align}  
We notice that [see Eq.~\eqref{Jpto1}, \eqref{HW-1},
\eqref{eq-functionchisemigroup}, and Definition~\ref{Delta}]
\begin{align}\label{NC}
\qHH[\uw]_{\bchi_\alpha, H_\ph} 
\ = \ & 
H_\ph \, + \, \bchi_\alpha \big( \qHH[\uw] - H_\ph \big) \bchi_\alpha 
\nonumber \\[1ex]
\ = \ & 
H_\ph \, \chi^2_\alpha \, + \, w_{0,0}(H_\ph) \, \bchi^2_\alpha 
\, + \, \bchi_\alpha \, \qWW[\uw] \, \bchi_\alpha
\nonumber \\[1ex]
\ = \ & 
\Delta_{\chi_\alpha, H_\ph}[w_{0,0}(H_\ph)] \, + \, 
\bchi_\alpha \, \qWW[\uw] \, \bchi_\alpha \, . 
\end{align}  
A norm-convergent Neumann series expansion
together with \eqref{NYC1} and \eqref{NC} imply that
$\qHH[\uw]_{\bchi_\alpha, H_\ph}$ is invertible and
\begin{align} \label{Neum}
(H & [\uw]_{\bchi_\alpha, H_\ph} )^{-1} \ = \ 
\\ \nonumber &
\sum_{n=0}^\infty (-1)^n \Delta_{\chi_\alpha, H_\ph}[w_{0,0}(H_\ph)]^{-1} 
\Big( \bchi_\alpha \, \qWW[\uw] \, \bchi_\alpha \,
\Delta_{\chi_\alpha, H_\ph}[w_{0,0}(H_\ph)]^{-1} \Big)^n \, .  
\end{align}  
Eqs.~\eqref{NYC1} and \eqref{Neum} imply that 
\begin{align}
\big\| ( \qHH[\uw]_{\bchi_\alpha, H_\ph} )^{-1} \big\|_{\mathcal{B}(\fH_{\bchi_\alpha})} 
\ \leq \
\frac{1}{  1 -   \frac{4 e^{\alpha} \xi \| \uw_\gI \|^{(\xi)} 
}{1 - 2 \|\uw - \ur\|_{(\partial_r)}} } \, 
\frac{4 e^{\alpha}}{1 - 2\|\uw - \ur\|_{(\partial_r)}} \, , 
\end{align}  
which yields \eqref{H-1}. 

Now we prove the last statement of the lemma, namely,
the invertibility of $\qHH[\uw(z)]$ under the condition that
$\rRe(z) \geq \xi \| \uw_\gI \|^{ (\xi)} 2 e^{\alpha}$.
We proceed as in \eqref{EM1} and obtain 
\begin{align}\label{NYC2}
\rRe\{ w_{0,0}(z,r) \} 
\ = \ &
\rRe\{z\} \, + \, r \, + \, \rRe\{ w_{0,0}(z,r) - w_{0,0}(z,0) - r \}
\nonumber \\[1ex]
\ \geq \ & 
\rRe\{z\} \, + \, r \big( 1 -  \|  \uw  - \ur  \|_{(\partial_r)} \big) 
\ \geq \ \rRe\{z\} \, , 
\end{align}  
uniformly in $0 \leq r \leq 1$. We conclude as in
\eqref{NYC1} that
\begin{align}\label{NYC4}
\big\| \qWW[\uw] \, w_{0,0}(H_\ph)^{-1} \big\|_\op 
\ \leq \ 
\frac{\xi \, \| \uw_\gI \|^{(\xi)} \, 2 e^{\alpha}}{ \rRe(z) } 
\ < \ 1 \, ,
\end{align}  
which establishes $0 \in \rho( \qHH[\uw] )$, since a Neumann series
expansion as in \eqref{Neum} proves the invertibility of 
$\qHH[\uw] = w_{0,0}(H_\ph) + \qWW[\uw]$.
\end{proof}
\end{lemma}
%
We formulate Lemma~\ref{lemma: domFchiconditions} under
stronger and simpler hypotheses which yield simpler norm
bounds, too.
%
\begin{cor} \label{cor:domFchiconditions}
Suppose that $|z | \leq \frac{1}{4} e^{-\alpha}$, that 
$\|\uw - \ur\|_{(\partial_r)} < \frac{1}{10}$, and that
$\| \uw_\gI \|^{(\xi)} \leq \frac{e^{-\alpha}}{100}$. 
Then, $\Delta_{\chi_{\alpha }, H_\ph}[w_{0,0}(H_\ph)]\big|_{\fH_{\bchi_\alpha}}$
and $H[\uw]_{\bchi_\alpha, H_\ph}\big|_{\fH_{\bchi_\alpha}}$ are invertible and
these inverses obey the norm bounds
\begin{align}\label{H-0a}
\Big\| \big( 
\Delta_{\chi_{\alpha + \beta}, H_\ph}[w_{0,0}(H_\ph)] \big|_{\fH_{\bchi_\alpha}} 
\big)^{-1} \Big\|_{\mathcal{B}(\fH_{\bchi_\alpha})}
\ \leq \ 
5 \, e^{\alpha}
\end{align}  
and
\begin{align}\label{H-1a}
\big\| ( H[\uw]_{\bchi_\alpha, H_\ph} )^{-1} \big\|_{\mathcal{B}(\fH_{\bchi_\alpha})} 
\ \leq \ 
10 \, e^{\alpha} \, . 
\end{align}  
Moreover, $H[\uw] \equiv H[\uw(z)]$ is invertible provided 
$\rRe(z) \geq \xi \| \uw_\gI \|^{ (\xi)} 2 e^{\alpha}$. 
\end{cor}

\begin{theorem}\label{thm-specrg}
Let $0< r_Z < \frac{1}{4}$, and assume that 
$\uw \in \cW_Z^\xi$ [see Eqs.~\eqref{wmn-05}-\eqref{wmn-06,2}] and that  
\begin{align} \label{cond-specrg-01a}
(10 e^\alpha)^2 \, \big( \xi \, \|\uw_\gI\|_Z^{(\xi)} \big)^2 & \, 
\big( 2 + \|\uw\|_Z^{(\xi) } \big) 
\ < \ \frac{1}{8} 
\quad \text{and} 
\\[1ex] \label{cond-specrg-01b}
( \xi \, \| \uw_\gI \|_Z^{(\xi)} )^2 
& \ < \ 
\frac{e^{-2\alpha}}{20} \, \big( \tfrac{1}{4}  - r_Z \big) \, .
\end{align}  
Moreover suppose that \eqref{H-1a} is satisfied. Then, the
following statements hold true:
\begin{itemize}
\item[(a)] The function 
$Q_\alpha: D(\tfrac{1}{4} e^{-\alpha}) \cap Q_\alpha^{-1}[D(r_Z)] 
\to D(r_Z)$ [see \eqref{Qgamma}] is biholomorphic. 

\item[(b)] The complex derivative of $Q_\alpha^{-1}$ fulfills
the norm estimate
\begin{align}
\big| \partial_\zeta Q^{-1}_\alpha[\zeta] \big| 
\ \leq \ 
\frac{e^{-\alpha}}{1 - (10 e^{\alpha})^2  
(\xi \, \|\uw_\gI\|_Z^{(\xi)})^2 (2 + \|\uw\|_Z^{(\xi)}) } 
\ \leq \ 2 \, . 
\end{align}  

\item[(c)] Additionally assuming that 
$10 e^{\alpha}  ( \xi \, \|\uw_\gI\|_Z^{(\xi)} )^2 < e^{-\alpha}r_Z$, 
it follows that 
\begin{align}\label{contain}
D\Big(e^{-\alpha} r_Z - 10 e^{\alpha} (\xi \| \uw_\gI \|_Z^{(\xi)})^2 \Big)
\ \subset \ &
D\big( \tfrac{1}{4} e^{-\alpha} \big) \, \cap \, 
Q_\alpha^{-1}[D(r_Z)] 
\nonumber \\
\ = \ &
\boE_\alpha(D_{r_Z}) \, , 
\end{align}  
where we recall Definition~\ref{E}.  
\end{itemize}
\begin{proof}[Proof of (a)]
For $\zeta \in D(r_Z)$ and $|z| \leq \tfrac{1}{4} e^{-\alpha}$, 
we introduce the function
\begin{align} \label{hzdef}
h(z) \ := \ 
z \, + \, e^{-\alpha} \, \zeta \, - \, e^{-\alpha} \, Q_\alpha(z) \, .
\end{align}  
Notice that $\chi_\alpha(0) = 1$ and $w_{0,0}(z,0) = z$, [see
\eqref{eq-functionchisemigroup}, and \eqref{w00z}] and $H_\ph \Om =0$.
Moreover, we recall Definition~\ref{def:rescaledfsmap} and
Eq.~\eqref{Sal} and note that $\Gamma_\alpha \Om = \Om$.  Using
Eqs.~\eqref{FchiT}, \eqref{Hchi} and that $W_{H_\ph} = H[\uw] - H_\ph$
[see \eqref{HTW}], we obtain:
\begin{align} \label{hz}
h(z) \ = \ & 
z + e^{-\alpha} \zeta - \big\la F_{\chi_\alpha, H_\ph}(H[\uw]) \big\ra_\Om
\nonumber \\[1ex]
\ = \ & 
z + e^{-\alpha}\zeta - \la H[\uw]_{\chi_\alpha, H_\ph} \ra_\Om 
\nonumber \\  &           
+ \big\la \chi_\alpha \big( H[\uw] - H_\ph \big) \, \bchi_\alpha 
\big( H[\uw]_{\bchi_\alpha, H_\ph} \big)^{-1} \bchi_\alpha 
\big( H[\uw] - H_\ph \big) \chi_\alpha \big\ra_{\Om}
\nonumber \\[1ex]
\ = \ & 
z + e^{-\alpha}\zeta - z  + 
\big\la \qWW[\uw] \, \bchi_\alpha ( H[\uw]_{\bchi_\alpha, H_\ph})^{-1} 
\bchi_\alpha \, \qWW[\uw] \, \big\ra_\Om
\nonumber \\[1ex]
\ = \ & 
e^{-\alpha}\zeta + \big\la \qWW[\uw] \, \bchi_\alpha (H[\uw]_{\bchi_\alpha, H_\ph})^{-1}
\bchi_\alpha \, \qWW[\uw] \, \big\ra_\Om \, ,
\end{align}  
where we recall that the vacuum expectation value 
$\la A \ra_{\Om}$ of an operator $A$ is defined as 
$\la A \ra_{\Om} := \la \Om | \, A \Om \ra$.
We observe that, thanks to \eqref{H-1a} and \eqref{cond-specrg-01b},
\begin{align} \label{cobtract-a}
\big| h(z)- e^{-\alpha }\zeta \big| 
\ \leq \ &
\| \qWW[\uw] \|_\op^2 \, 
\big\| ( H[\uw]_{\bchi_\alpha, H_\ph} )^{-1} \big\|_{\cB(\fH_{\bchi_\alpha})}^2  
\nonumber \\[1ex] 
\ \leq \ &
10 e^{\alpha} \, \big( \xi \, \| \uw_\gI \|_Z^{(\xi)}  \big)^2 
\ \leq \
\tfrac{1}{2} \, e^{-\alpha } ( \tfrac{1}{4} - r_Z) \, , 
\end{align} 
which implies that 
\begin{align} \label{cobtract-b}
h\big( \bD( \tfrac{1}{4} \, e^{-\alpha } ) \big)
\ \subseteq \ 
\bD\big( \tfrac{1}{2} \, ( \tfrac{1}{4} + r_Z) \, e^{-\alpha } \big) 
\ \subseteq \ 
\bD( \tfrac{1}{4} \, e^{-\alpha } ) \, .
\end{align} 
Next, we calculate
\begin{align}
|\partial_z h(& z)| 
\ \leq \ 
\big| \partial_z \big\la \qWW[\uw] \, \bchi_\alpha 
(H[\uw]_{\bchi_\alpha, H_\ph})^{-1} \bchi_\alpha \, \qWW[\uw] \big\ra_\Om \big|
\nonumber \\[1ex]
\ \leq \ & 
\big| \big\la W[\partial_z \uw] \, \bchi_\alpha (H[\uw]_{\bchi_\alpha, H_\ph})^{-1}
\bchi_\alpha \, \qWW[\uw] \big\ra_\Om \big| 
\\ \nonumber &
+ \big| \big\la \qWW[\uw] \, \bchi_\alpha (H[\uw]_{\bchi_\alpha, H_\ph})^{-1}
\bchi_\alpha \, W[\partial_z\uw] \big\ra_\Om \big| 
\\ \nonumber & 
+ \big| \big\la \qWW[\uw] \, \bchi_\alpha (H[\uw]_{\bchi_\alpha, H_\ph})^{-1}
\, H[\partial_z \uw]_{\bchi_\alpha, H_\ph} \, (H[\uw]_{\bchi_\alpha, H_\ph})^{-1} 
\bchi_\alpha \, \qWW[\uw] \big\ra_\Om \big| \, .
\end{align}  
Estimating the vacuum expectation values by the corresponding operator
norms, we get [see \eqref{opbound}]:
\begin{align} \label{partailh}
|\partial_z h(z)| 
\ \leq \ & 
2 \| \qWW[\uw] \|_\op \, \| W[\partial_z\uw] \|_\op \, 
\big\| ( H[\uw]_{\bchi_\alpha, H_\ph} )^{-1} \big\|_{\cB(\fH_{\bchi_\alpha})}
\nonumber \\ &
+ \| \qWW[\uw] \|_\op^2 \, 
\big\| ( H[\uw]_{\bchi_\alpha, H_\ph} )^{-1} \big\|_{\cB(\fH_{\bchi_\alpha})}^2 \, 
\| H[\partial_z \uw]_{\bchi_\alpha, H_\ph} \|_\op
\nonumber \\[1ex]
\ \leq \ & 
\big( 10 e^{\alpha} \big)^2 \, 
\big( \xi \, \| \uw_\gI \|_Z^{(\xi)} \big)^2 \, 
\big( 2 + \| \uw \|_Z^{(\xi) } \big)  
\ \leq \ 
\frac{1}{8} \, ,
\end{align} 
where we use \eqref{opbound}, \eqref{wmn-06,3}, and \eqref{H-1a} to
derive the second and \eqref{cond-specrg-01a} to derive the third
inequality. It follows for every $z_1, z_2 \in \CC$ with 
$|z_1|, |z_2| \leq \frac{1}{4} e^{-\alpha}$, that
\begin{align} \label{cobtract}
|h(z_1) - h(z_2)|  
\ \leq \ 
(10 e^\alpha)^2 \, \big( \xi \, \| \uw_\gI \|^{(\xi)} \big)^2 \, 
\big( 2 + \| \uw  \|_Z^{( \xi) } \big) \, |z_1-z_2| \, . 
\end{align}  
This and Eq.~\eqref{cobtract-b} imply that $h$ defines a contraction
on $\bD( \tfrac{1}{4} \, e^{-\alpha })$. By the contraction mapping
principle, $h$ possesses a unique fixed point 
$z_\zeta \in \bD( \tfrac{1}{4} \, e^{-\alpha })$, and by
\eqref{cobtract-b}, we can strengthen this estimate to 
$z_\zeta \in \bD( \tfrac{1}{2} [\tfrac{1}{4} - r_z] e^{-\alpha })$.
Thanks to the fixed point equation $z_\zeta = h(z_\zeta)$, we have
that $Q_\alpha(z_\zeta) = \zeta$. Hence, $Q_\alpha$ is surjective, and
the uniqueness of the fixed point implies the injectivity and hence
\textit{(a)}.\\[-1,5ex]

\noindent\textit{Proof of (b).}
Using $Q_\alpha(z) = \zeta + e^\alpha [z - h(z)]$, which
implies that $\partial_z Q_\alpha = e^\alpha [1 - \partial_z h(z)]$ 
in combination with \eqref{partailh}, we estimate the complex
derivative of $Q_\alpha^{-1}$ as follows,
\begin{align} \label{eq-tele.1}
\big| \partial_\zeta [Q_\alpha^{-1}] (\zeta) \big| 
\ = \ & 
\big| \partial_z Q_\alpha[ Q_\alpha^{-1} (\zeta) ] \big|^{-1}
\ = \
e^{-\alpha} \, \big| 1 \, - \, 
\partial_z h[ Q_\alpha^{-1} (\zeta) ] \big|^{-1}
\nonumber \\[1ex]
\ \leq \ &
\frac{e^{-\alpha}}{1 - 10 e^{-\alpha} 
(\xi \, \|\uw_\gI\|_Z^{(\xi)})^2 (2 + \|\uw\|_Z^{(\xi)}) } \, .
\end{align}  

\noindent\textit{Proof of (c).}
It follows from \eqref{hz} and \eqref{hzdef} that 
\begin{align} \label{qzmz-a}
|Q_\alpha(z)- e^{\alpha} z | 
\ = \ &
e^{\alpha}
\big| \big\la \qWW[\uw] \, \bchi_\alpha ( H[\uw]_{\bchi_\alpha, H_\ph} )^{-1}
\bchi_\alpha \, \qWW[\uw] \big\ra_\Om \big|  
\nonumber \\[1ex]
\ \leq \ & 
10 e^{2\alpha} \, ( \xi \, \|\uw_\gI\|^{(\xi)} )^2 \, ,
\end{align}  
where we use \eqref{H-1a} and \eqref{opbound}. With this we obtain
\begin{align}\label{qzmz-b}
|Q_\alpha(z)| \ \leq \ 
e^{\alpha } \, |z| \, + \, 
10 e^{2\alpha} \, ( \xi \, \|\uw_\gI\|^{(\xi)} )^2 \, .
\end{align}  
This shows that $Q_\alpha(z) \in D_{r_Z}$, whenever 
$|z| < e^{-\alpha} r_Z - 10 e^{\alpha} ( \xi \| \uw_\gI \|^{(\xi)} )^2$.
\end{proof}
\end{theorem}

\begin{bem}\label{bem}
Assume that $0< r_Z < \frac{1}{4}$ and $\uw \in \cW_Z^\xi$, 
Eqs.~\eqref{cond-specrg-01a}, \eqref{cond-specrg-01b}, 
\eqref{H-1a}, and 
\begin{align} \label{hypo}
10 e^{\alpha} \, ( \xi \, \|\uw_\gI\|^{(\xi)} )^2 \, + \, 
2 e^{\alpha} \, \xi \, \| \uw_\gI \|^{ (\xi)} 
\ < \ e^{-\alpha} \, r_Z \, .
\end{align}  
Furthermore set 
\begin{align} \label{open} 
\cA_{\uw, \alpha} 
\ := \ &
D\big( e^{-\alpha}r_Z - 10 e^{\alpha} (\xi \|\uw_\gI\|^{(\xi)} )^2 \big)
\: \cap \: 
\big\{ z \in \CC : \ \rRe(z) > 2 e^{\alpha} \xi \|\uw_\gI\|^{ (\xi)} \big\} 
\nonumber \\
\ \neq \ & \emptyset \, .
\end{align}  
Then $\cA_{\uw, \alpha} \subset \boE_\alpha(D_{r_Z})$ 
[see \eqref{contain}] and $H[\uw(z)]$ is bounded invertible 
for every $z \in \cA_{\uw, \alpha}$ [see the text below \eqref{H-1}]. 
The definition of $\cR_\alpha$,  
\begin{align}
\cR_\alpha\big( H[\uw] \big)(\zeta) 
\ = \ 
\hcR_\alpha\big( H[\uw(\boE_\alpha(\zeta))] \big),  
\end{align}  
implies that $\cR_\alpha\big( H[\uw] \big)(\zeta)$ is bounded invertible 
for $\zeta$ belonging to the open set 
$\boE_\alpha^{-1}( \cA_{\uw, \alpha} )$, see Proposition~\ref{prop-isospecSFB}.   
\end{bem}

\subsection{The Flow Property} \label{subsec-flowprop}
%
\begin{theorem}\label{thm-fullsemigroupprima}
Suppose that $\uw$ and $\alpha$ satisfy the hypotheses of
Lemma~\ref{lemma: domFchiconditions},
Corollary~\ref{cor:domFchiconditions}, and
Theorem~\ref{thm-specrg}. Suppose furthermore that there exist 
$\tuw \in \cW^\xi $, such that $H[\tuw] = \cR_\alpha(H[\uw]) $ and
$\tuw$ and $\beta$ satisfy the same hypothesis as $\uw$ and $\alpha$
together with \eqref{hypo}. Then
\begin{align}\label{uno}
\boE_{\alpha + \beta, \uw} 
\ = \ 
\boE_{\alpha, \uw} \circ \boE_{\beta, \tuw} \, ,      
\end{align}  
and
\begin{align} \label{dos}
\forall \, z\in D(r_Z): \quad
\cR_\beta \circ \cR_\alpha (H[\uw])(z) 
\ = \ 
\cR_{\alpha+\beta}(H[\uw])(z) \, .
\end{align}  
\begin{proof}
Remark~\ref{bem} implies that 
\begin{align}\label{lelele}
\cR_{\beta}\big( H[\tuw] \big) (z)
\ = \ 
\cR_{\beta}\big[ \cR_\alpha(\uw) \big](z) 
\ = \
\hcR_{\beta}\big[ \hcR_\alpha\big( 
\uw[ \boE_{\alpha, \uw} \circ \boE_{\beta, \tuw}(z) ] \big) \big]
\end{align}  
is bounded invertible, for $z$ in the open set 
$\boE_{\beta}^{-1}( \cA_{\tuw, \beta} )$. It follows from 
Theorem~\ref{thm-semigroup} that 
\begin{align} \label{eq-flowprop}
\hcR_{\beta}\big[ \hcR_\alpha\big( 
\uw[ \boE_{\alpha, \uw} \circ \boE_{\beta, \tuw}(z) ] \big) \big]
\ = \ 
\hcR_{\alpha+\beta}\big( 
\uw[ \boE_{\alpha, \uw} \circ \boE_{\beta, \tuw}(z)] \big) \, . 
\end{align}  
Using Definitions~\ref{def-Qgamma} and \ref{E} together with 
\eqref{eq-vacexpz} we get that 
\begin{align}
Q_{\alpha+\beta}[ \boE_{\alpha, \uw} \circ \boE_{\beta, \tuw}(z) ] 
\ =  \ z \, , 
\end{align} 
which implies that 
$\boE_{\alpha, \uw} \circ \boE_{\beta, \tuw}(z) = 
Q_{\alpha+\beta}^{-1}(z) = \boE_{\alpha+\beta, \uw}(z)$, see
Definition~\ref{E}. This leads us to
\begin{align}\label{europ}
\boE_{\alpha + \beta, \uw} (z) 
\ = \  
\boE_{\alpha, \uw} \circ \boE_{\beta, \tuw}(z) \, ,
\end{align}  
for every $z$ in $\boE_{\beta}^{-1}( \cA_{\tuw, \beta} )$.  Since the
functions in \eqref{europ} are analytic, the same holds true for every
$z \in D(r_Z)$, which yields \eqref{uno}. Now, Eq.~\eqref{uno}
together with \eqref{lelele} and \eqref{eq-flowprop} imply that
\begin{align}
\forall\, z \in \boE_{\beta}^{-1}(\cA_{\tuw, \beta}): \quad
\cR_\beta \circ \cR_\alpha\big( H[\uw] \big)(z) 
\ = \ 
\cR_{\alpha+\beta}\big( H[\uw] \big)(z) \, ,
\end{align}  
and, therefore, an analyticity argument leads us to \eqref{dos}.  
\end{proof}
\end{theorem}

\subsection*{Acknowledgement}
Research supported by CONACYT, FORDECYT–PRONACES 429825/2020 (proyecto
apoyado por el FORDECYT–PRONACES, PRONACES/429825), recently renamed
project CF-2019 / 429825. This work was supported by the project
PAPIIT–DGAPA–UNAM IN114925. M.~B.\ is a Fellow of the Sistema
Nacional de Investigadores (SNI).


\begin{thebibliography}{10}

\bibitem{BachBallesterosFroehlich2015}
V.~Bach, M.~Ballesteros, and J.~Fr{\"o}hlich.
\newblock Continuous renormalization group analysis of spectral problems in
  quantum field theory.
\newblock {\em J.~Functional Analysis}, 268(4):749--823, 2015.

\bibitem{BachChenFroehlichSigal2003}
V.~Bach, T.~Chen, J.~Fr{\"{o}}hlich, and I.~M. Sigal.
\newblock Smooth {F}eshbach map and operator-theoretic renormalization group
  methods.
\newblock {\em J.~Funct.~Anal.}, 203(1):44--92, 2003.

\bibitem{BachChenFroehlichSigal2007}
V.~Bach, T.~Chen, J.~Fr{\"{o}}hlich, and I.~M. Sigal.
\newblock The renormalized electron mass in non-relativistic quantum
  electrodynamics.
\newblock {\em J.~Func.~Anal.}, 243(2):426--535, February 2007.

\bibitem{BachFroehlichSigal1995}
V.~Bach, J.~Fr{\"{o}}hlich, and I.~M. Sigal.
\newblock Mathematical theory of non-relativistic matter and radiation.
\newblock {\em Lett.~Math.~Phys.~}, 34:183--201, 1995.

\bibitem{BachFroehlichSigal1998a}
V.~Bach, J.~Fr{\"{o}}hlich, and I.~M. Sigal.
\newblock Quantum electrodynamics of confined non-relativistic particles.
\newblock {\em Adv.~in Math.~}, 137:299--395, 1998.

\bibitem{BachFroehlichSigal1998b}
V.~Bach, J.~Fr{\"{o}}hlich, and I.~M. Sigal.
\newblock Renormalization group analysis of spectral problems in quantum field
  theory.
\newblock {\em Adv.~in Math.~}, 137:205--298, 1998.

\bibitem{BallesterosFaupinFroehlichSchubnel2015}
M.~Ballesteros, J.~Faupin, J.~Fr{\"o}hlich, and B.~Schubnel.
\newblock Quantum electrodynamics of atomic resonances.
\newblock {\em Commun.~Math.~Phys.}, 337:633--680, 2015.

\bibitem{Chen2008}
T.~Chen.
\newblock Infrared renormalization in nonrelativistic {QED} and scaling
  criticality.
\newblock {\em J.~Funct.\ Analysis}, 254:2555--2647, 2008.

\bibitem{Faupin2008}
J.~Faupin.
\newblock Resonances of the confined hydrogen atom and the {L}amb-{D}icke
  effect in nonrelativistic {QED}.
\newblock {\em Ann.\ Henri Poincaré}, 9:743--773, 2008.

\bibitem{Feshbach1958}
H.~Feshbach.
\newblock Unified theory of nuclear reactions.
\newblock {\em Ann.~Phys.}, 5:357--390, 1958.

\bibitem{FroehlichGriesemerSigal2011}
J.~Fr{\"o}hlich, M.~Griesemer, and I.~Sigal.
\newblock Rev.~math.~phys.
\newblock {\em Spectral renormalization group and local decay in the standard
  model of non-relativistic quantum electrodynamics}, 23(2):179--209, 2011.

\bibitem{GriesemerHasler2008}
M.~Griesemer and D.~Hasler.
\newblock On the smooth {F}eshbach-{S}chur map.
\newblock {\em Journal of Functional Analysis}, 254(9):2329--2335, 2008.

\bibitem{GriesemerHasler2009}
M.~Griesemer and D.~Hasler.
\newblock Analytic perturbation theory and renormalization analysis of matter
  coupled to quantized radiation.
\newblock {\em Ann.~Henri Poincar{\'e}}, 10:577--621, 2009.

\bibitem{Grushin1971}
V.~Grushin.
\newblock Les problèmes aux limites dégénérés et les opérateurs
  pseudodifférentiels.
\newblock {\em Actes du Congrès International des Mathématiciens (Nice,
  1970)}, 2:737--743, 1971.

\bibitem{HaslerHerbst2011}
D.~Hasler and I.~Herbst.
\newblock Ground states in the spin boson model.
\newblock {\em Ann.\ Henri Poincar{\'e}}, 12:621--677, 2011.

\bibitem{HaslerHerbstHuber2008}
D.~Hasler, I.~Herbst, and M.~Huber.
\newblock On the lifetime of quasi-stationary states in nonrelativistic {QED}.
\newblock {\em Ann.\ Henri Poincar{\'e}}, 9:1005--1028, 2008.

\bibitem{Schur1917}
I.~Schur.
\newblock {\"U}ber {P}otenzreihen, die im {I}nnern des {E}inheitskreises
  beschr{\"a}nkt sind.
\newblock {\em J.~Reine Angew. Math.}, 147:205--232, 1917.

\bibitem{Sigal2009}
I.~Sigal.
\newblock Ground state and resonances in the standard model of the
  non-relativistic {QED}.
\newblock {\em J.~Stat.\ Phys.}, 134(5-6):899--939, 2009.

\end{thebibliography}

\end{document}